\documentclass[onecolumn,draftcls,12pt]{IEEEtran} \def\mydocumentclass{1} 

\usepackage[noadjust]{cite}

\ifCLASSINFOpdf
  \usepackage[pdftex]{graphicx}
\else
  \usepackage[dvips]{graphicx}
\fi
\graphicspath{{./figs/}}

\usepackage[T1]{fontenc} 
\usepackage{mathtools,amsmath,amssymb,amsthm}
\usepackage[cmintegrals]{newtxmath}
\usepackage{bm} 
 
\usepackage{algorithmic}

\usepackage{array}

\ifCLASSOPTIONcompsoc
  \usepackage[caption=false,font=normalsize,labelfont=sf,textfont=sf]{subfig}
\else
  \usepackage[caption=false,font=footnotesize]{subfig}
\fi

\usepackage[hyphens]{url}

\usepackage{adjustbox}
\usepackage{makecell} 
\usepackage{multirow}
\usepackage{supertabular}
\usepackage[para]{threeparttable}
\usepackage{mleftright}
\mleftright

\usepackage[ruled,linesnumbered,vlined]{algorithm2e}
\usepackage{etoolbox}
\makeatletter
\patchcmd{\@algocf@start}
  {-1.5em}
  {0pt}
  {}{}
\makeatother

\newtheorem{theorem}{Theorem}
\newtheorem{proposition}{Proposition}

\theoremstyle{definition}
\newtheorem{definition}{Definition}
\theoremstyle{remark}
\newtheorem*{remark}{Remark}

\usepackage[binary-units]{siunitx}
\DeclareSIUnit{\belmilliwatt}{Bm}
\DeclareSIUnit{\dBi}{\deci\bel{}i}
\DeclareSIUnit{\dBm}{\deci\belmilliwatt}
\DeclareSIUnit{\bps}{bps}

\DeclareMathOperator*{\maximize}{maximize}
\DeclareMathOperator*{\minimize}{minimize}
\DeclareMathOperator*{\argmax}{argmax}

\DeclareMathOperator*{\subjto}{\textnormal{subject to}}

\newcommand{\bp}{\ensuremath{{\mathbf p}}}
\newcommand{\bq}{\ensuremath{{\mathbf q}}}

\newcommand{\blambda}{\ensuremath{{\bm \lambdaup}}}
\newcommand{\bv}{\ensuremath{{\mathbf v}}}

\newcommand{\barbp}{\ensuremath{{\bar \bp}}}
\newcommand{\barbq}{\ensuremath{{\bar \bq}}}

\newcommand{\calN}{\ensuremath{{\mathcal N}}}

\newcommand{\Pmax}{\ensuremath{{P_{\textnormal{max}}}}}
\newcommand{\Pmaxk}{\ensuremath{{P_{\textnormal{max},k}}}}

\newcommand{\Pone}{\ensuremath{(\textnormal{P1})}}
\newcommand{\Ptwo}{\ensuremath{(\textnormal{P2})}}
\newcommand{\Pt}{\ensuremath{(\textnormal{P}_1^t)}}
\newcommand{\Pteq}{\ensuremath{(\textnormal{P}_2)}}
\newcommand{\Slave}{\ensuremath{(\textnormal{S}_{k})}}
\newcommand{\Master}{\ensuremath{(\textnormal{M})}}
\newcommand{\Q}{\ensuremath{(\textnormal{Q}_k)}}
\newcommand{\Qone}{\ensuremath{(\textnormal{Q}_{k,1}^{\varphi_k})}}
\newcommand{\Qtwo}{\ensuremath{(\textnormal{Q}_{k,2}^{\varphi_k})}}
\newcommand{\Qthree}{\ensuremath{(\textnormal{Q}_{k,3}^{\varphi_k})}}
\newcommand{\Qfour}{\ensuremath{(\textnormal{Q}_{k,4}^{\varphi_k})}}
\newcommand{\Dual}{\ensuremath{(\textnormal{D})}}
\newcommand{\Dualh}{\ensuremath{(\textnormal{D}^\bh)}}

\newcommand{\bh}{\ensuremath{{\mathbf h}}}

\newcommand{\bx}{\ensuremath{{\mathbf x}}}

\newcommand{\bP}{\ensuremath{{\mathbf P}}}

\newcommand{\bR}{\ensuremath{{\mathbf R}}}

\newcommand{\E}{\ensuremath{{\mathbb E}}}

\newcommand{\suchthat}{%
  \nonscript\;
  \ifnum\currentgrouptype=16
    \middle|
  \else
    \@suchthat
  \fi
  \nonscript\;
}

\begin{document}

\title{Low-Complexity Dynamic Resource Scheduling for Downlink MC-NOMA Over Fading Channels}

\author{Do-Yup~Kim,~\IEEEmembership{Graduate Student Member,~IEEE,}
        Hamid~Jafarkhani,~\IEEEmembership{Fellow,~IEEE,}
        and~Jang-Won~Lee,~\IEEEmembership{Senior Member,~IEEE}
\if\mydocumentclass1
	\thanks{This work has been submitted to the IEEE for possible publication. Copyright may be transferred without notice, after which this version may no longer be accessible.}%
\fi
\thanks{A preliminary version of this work~\cite{kim2021low} will be presented at IEEE VTC2021-Spring, in which only single-channel transmission has been taken into account, and neither subchannel scheduling nor inter-channel power scheduling has been dealt with.}
\thanks{D.-Y. Kim and J.-W. Lee are with the Department of Electrical and Electronic Engineering, Yonsei University, Seoul 03722, South Korea (e-mail: danny.doyup.kim@yonsei.ac.kr; jangwon@yonsei.ac.kr).}
\thanks{H. Jafarkhani is with the Center for Pervasive Communications and Computing, Department of Electrical Engineering and Computer Science, University of California at Irvine, Irvine, CA 92697, USA (e-mail: hamidj@uci.edu).}}


\maketitle

\begin{abstract}
In this paper, we investigate dynamic resource scheduling (i.e., joint user, subchannel, and power scheduling) for downlink multi-channel non-orthogonal multiple access (MC-NOMA) systems over time-varying fading channels. Specifically, we address the weighted average sum rate maximization problem with quality-of-service (QoS) constraints. In particular, to facilitate fast resource scheduling, we focus on developing a very low-complexity algorithm. To this end, by leveraging Lagrangian duality and the stochastic optimization theory, we first develop an opportunistic MC-NOMA scheduling algorithm whereby the original problem is decomposed into a series of subproblems, one for each time slot. Accordingly, resource scheduling works in an online manner by solving one subproblem per time slot, making it more applicable to practical systems. Then, we further develop a heuristic joint subchannel assignment and power allocation (Joint-SAPA) algorithm with very low computational complexity, called Joint-SAPA-LCC, that solves each subproblem. Finally, through simulation, we show that our Joint-SAPA-LCC algorithm provides good performance comparable to the existing Joint-SAPA algorithms despite requiring much lower computational complexity. We also demonstrate that our opportunistic MC-NOMA scheduling algorithm in which the Joint-SAPA-LCC algorithm is embedded works well while satisfying given QoS requirements.
\end{abstract}


\section{Introduction}
\label{sec:intro}

With the exponential proliferation of mobile devices, overall mobile data traffic is expected to grow to $77$ exabytes per month by $2022$~\cite{forecast2019cisco}.
Such surge in mobile data traffic will exacerbate resource shortages, which in turn will necessitate high levels of connectivity~and~spectral efficiency.
In these circumstances, non-orthogonal multiple access (NOMA) has been envisioned as a promising technology for future cellular networks thanks to its potential to achieve high connectivity and high spectral efficiency compared to orthogonal multiple access (OMA)-based technologies~\cite{dai2018survey, seo2019high, maraqa2020survey, makki2020survey}.
Especially, in~\cite{wei2020performance}, the potential performance gains of NOMA over OMA have been extensively studied in various system setups.
Unlike OMA, which essentially excludes inter-user interference (IUI), NOMA is based on the premise that IUI is acceptable, and two popular categories of NOMA include power-domain NOMA and code-domain NOMA~\cite{dai2015non}.
This paper focuses on the power-domain NOMA that multiplexes multiple users on the same resource based on superposition coding in the power domain, and then, at the receiver, multi-user detection is realized by mitigating IUI based on successive interference cancellation (SIC) (see \cite{kazemitabar2009performance, ganji2016interference, ganji2021asynchronous} and references therein).
However, despite the high network performance, the presence of IUI makes resource scheduling, which is very important in wireless networks, more difficult in NOMA systems.
In addition, resource scheduling techniques developed for the OMA systems, e.g., \cite{kim2012opportunistic, huang2009downlink, kim2014sum}, cannot be easily applied to the NOMA systems, and provide limited performance even if they can.
In this vein, resource scheduling with low computational complexity is one of the most paramount issues in NOMA systems, and thus, many studies have been conducted.
Nevertheless, they still have practical limitations, especially in terms of computational complexity.

\subsection{Related Work}
Early studies in this area have focused on single-channel NOMA (SC-NOMA).
Accordingly, various studies have been conducted in SC-NOMA systems, in terms of power allocation~\cite{xing2018optimal, wang2016power, choi2016power, yang2017optimality, wang2017convexity}, and power allocation and user selection~\cite{liu2018performance, datta2016optimal}.
More recently, the research focus in this area has been shifted from SC-NOMA to multi-channel NOMA (MC-NOMA).
MC-NOMA systems take multi-channel transmission into account; however, compared to SC-NOMA systems, resource allocation becomes much more complicated because of the additional burden of subchannel assignment.
As a result, the algorithms for SC-NOMA are usually inapplicable to MC-NOMA, and even if they are applicable, they provide limited performance.
Thus, to take full advantage of MC-NOMA, a new joint subchannel assignment and power allocation (Joint-SAPA) algorithm tailored to MC-NOMA systems is needed.

Joint-SAPA for minimizing total power consumption in MC-NOMA systems has been investigated thanks to the corresponding simple linear objective function~\cite{li2016dynamic, guo2019joint, yang2018power}.
Later, it has shown that the sum rate becomes a concave function of power allocation even though each user's data rate is a nonconcave function~\cite{fu2018subcarrier}.
Therefore, a basic power allocation problem to maximize the sum rate can be easily solved with well-known convex optimization solvers.
However, the convexity in optimization gets lost in more general Joint-SAPA problems that take into account subchannel assignment and practical constraints, e.g., a so-called \textit{SIC capacity constraint} that limits the number of users who can be served simultaneously through the same resource.
Hence, many heuristic Joint-SAPA algorithms have been proposed, e.g., \cite{fu2018subcarrier, cejudo2019fast, fu2019joint}, but almost all of them are still based on the concavity of the sum rate function with respect to power allocation.

Despite many studies on the sum rate maximization, the relative importance and/or fairness among users have not been addressed therein due to the nature of the sum rate performance metric.
Thereby, users with poor channel conditions may experience starvation because no resource might be allocated to them.
On the other hand, different tradeoffs can be achieved between the sum rate performance and the user fairness by controlling user weights in the weighted sum rate maximization problem.
However, unlike the (equally weighted) sum rate, the weighted sum rate is generally a nonconcave function of power allocation (even in SC-NOMA~\cite{wang2017convexity}), and accordingly the Joint-SAPA problem to maximize the weighted sum rate is known to be a strongly NP-hard problem~\cite{salaun2018optimal}.
Hence, in most cases, the ideas and underlying theory exploited in the sum rate maximization cannot be fully leveraged in the weighted sum rate maximization.
To address these, the Joint-SAPA problem to maximize the weighted sum rate has received much attention~\cite{liu2018dynamic, liu2020performance, saito2013system, parida2014power, sun2016optimal, di2016sub, lei2016power, salaun2020joint}.
In~\cite{liu2018dynamic}, the power allocation for each subchannel in a two-user MC-NOMA system has been investigated.
In~\cite{liu2020performance}, the Joint-SAPA problem in a multi-user MC-NOMA system has been investigated without considering the essential SIC capacity constraint.
In~\cite{saito2013system}, the authors have proposed a heuristic Joint-SAPA algorithm, considering the SIC capacity constraint, based on the fractional transmit power control (FTPC) and exhaustive search (ES) algorithms.
In~\cite{parida2014power, sun2016optimal}, heuristic Joint-SAPA algorithms using the difference-of-convex programming (DCP) approach have been developed under the assumption that each subchannel is occupied by up to two users.
In~\cite{di2016sub}, the power allocation and the subchannel assignment are performed based on the geometric programming (GP) approach and the many-to-many matching game, respectively.
In~\cite{lei2016power}, the authors have proposed a Joint-SAPA algorithm utilizing the Lagrangian dual and dynamic programming (DP) approaches.
Most recently, in~\cite{salaun2020joint}, the authors have studied a Joint-SAPA problem with further considering the individual subchannel power limits.
They have developed a Joint-SAPA algorithm based on the DP approach and the projected gradient descent (PGD) method.
The weighted sum rate maximization studies related to resource allocation for the MC-NOMA system are summarized in Table~\ref{table:comp}.

\begin{table*}[!t]
\if\mydocumentclass0
\renewcommand{\arraystretch}{1.3}
\fi
\centering
\caption{Comparison of studies on the weighted sum rate maximization in the MC-NOMA system (\checkmark: Considered)}\label{table:comp}
\begin{tabular}{c|cccc|ccc}
	\hline
	\multirow[b]{2}{*}{\bfseries{Ref.}} & \multicolumn{4}{c|}{\bfseries{Constraint}} & \multicolumn{3}{c}{\bfseries{Optimization}} \\
	\cline{2-8}
	& {\thead{Total\\power limit}} & {\thead{Subchannel\\power limit}} & {\thead{SIC\\capacity}} & {\thead{QoS\\requirement}} & {\thead{Subchannel\\assignment}} & {\thead{Power\\allocation}} & {\thead{Scheduling\\over fading}}\\
	\hline
	\cite{liu2018dynamic}
		& \checkmark & & & & & \checkmark & \\
	\cite{liu2020performance}
		& \checkmark & & & & \checkmark & \checkmark & \checkmark \\
	\cite{saito2013system, parida2014power, sun2016optimal, di2016sub, lei2016power}
		& \checkmark & & \checkmark & & \checkmark & \checkmark & \checkmark \\
	\cite{salaun2020joint}
		& \checkmark & \checkmark & \checkmark & & \checkmark & \checkmark & \\
	Our work
		& \checkmark & \checkmark & \checkmark & \checkmark & \checkmark & \checkmark & \checkmark \\
	\hline
\end{tabular}
\end{table*}

\subsection{Motivation and Contributions}
\label{subsec:motivation}
Joint-SAPA algorithms to maximize the weighted sum rate have been extensively studied in the literature.
However, all of them are still based on approaches that typically need high computational complexity (e.g., approaches based on FTPC and ES~\cite{saito2013system}, DCP~\cite{parida2014power, sun2016optimal}, GP and matching game~\cite{di2016sub}, and DP~\cite{lei2016power, salaun2020joint}).
Such high computational complexity will become increasingly burdensome for practical use in future cellular networks with very short time slots.\footnote{In recent standardization trends, the length of the slot, which is a unit to transmit $14$ orthogonal frequency division multiplexing (OFDM) symbols, is reduced to achieve higher spectral efficiency and traffic capacity and lower user plane latency. For example, 5G New Radio (NR) supports flexible OFDM numerology with subcarrier spacing from \SI{15}{\kilo\Hz} to \SI{240}{\kilo\Hz}, resulting in a slot length as short as \SI{62.5}{\micro\second}~\cite{dahlman20185g}. Even more, 5G NR introduces a unit of \textit{mini-slot}, which is even shorter than a slot, for the sake of fast data transmission for ultra-reliable low-latency communication (URLLC).}
Hence, Joint-SAPA algorithms with much lower computational complexity are needed to make it possible to generate transmit signals at the base station (BS) in a very short time slot.
%
%

In addition, not all weighted sum rate maximization studies have considered explicit QoS requirements, as shown in Table~\ref{table:comp}.
Instead, the authors in~\cite{liu2020performance, saito2013system, parida2014power, sun2016optimal, di2016sub, lei2016power} have realized proportional fair scheduling based on their own Joint-SAPA algorithms in simulation, using the fact that the proportional fair scheduling is a specific use case of the weighted sum rate maximization problem.
Although the proportional fair scheduling provides high sum rate performance while closing the performance gap between users to some extent, it cannot explicitly guarantee given QoS requirements.
Accordingly, in a practical QoS-aware system with individual user QoS requirements, a new scheduling technique that can meet the individual QoS requirements as well is needed.
In particular, in wireless network systems that are subject to time-varying fading channels, the development of a scheduling technique that meets QoS requirements by exploiting the variability of the channels is necessary.
Hence, in this paper, we aim to develop a novel low-complexity opportunistic resource scheduling algorithm for the downlink MC-NOMA system, which fully exploits the stochasticity of fading channels to maximize the weighted average sum rate while ensuring the individual QoS requirements of users.

The main contributions of this paper are summarized as follows:
\begin{itemize}
\item We address a dynamic resource scheduling problem for the downlink MC-NOMA system over time-varying fading channels. To the best of our knowledge, this is the first work to maximize the weighted average sum rate while ensuring explicitly given QoS requirements via joint optimization of user, subchannel, and power scheduling.
\item We develop a Joint-SAPA algorithm with very low computational complexity, called Joint-SAPA-LCC, to maximize the instantaneous weighted sum rate. It has much lower computational complexity compared to the existing Joint-SAPA algorithms with the same objective.
\begin{itemize}
\item[$\circ$] We prove that it is optimal to select up to two users per subchannel, assuming that the noise power of users suffering from interference is neglected, and propose a very simple optimal user selection rule based on it.
\item[$\circ$] In accordance with the proposed user selection rule, we derive closed-form optimal user power allocation formulas and a simple subchannel power allocation algorithm.
\item[$\circ$] Through simulation, we verify that our Joint-SAPA-LCC algorithm provides good performance comparable to the existing Joint-SAPA algorithms despite requiring much lower computational complexity.
\end{itemize}
\item By leveraging the Lagrangian duality and the stochastic optimization theory, we develop an opportunistic MC-NOMA scheduling algorithm that fully exploits time-varying fading channels. It operates in an online manner using the Joint-SAPA-LCC algorithm, and thus, it is very effective for practical use. Through simulation, we show that our opportunistic MC-NOMA scheduling works well and properly meets various QoS requirements.
\end{itemize}

\subsection{Paper Structure and Notations}
\subsubsection*{Paper Structure}
The rest of the paper is organized as follows.
In Section~\ref{sec:sys}, we formulate the system model and the dynamic resource scheduling problem.
In Section~\ref{sec:Joint-SAPA} and Section~\ref{sec:scheduling}, we develop the Joint-SAPA-LCC algorithm and the opportunistic MC-NOMA scheduling algorithm, respectively.
We present simulation results in Section~\ref{sec:sim} and conclude in Section~\ref{sec:conc}.

\subsubsection*{Notation}
Scalars, vectors, and sets are denoted by italic, boldface, and calligraphic letters, respectively.
A vector that consists of elements in the set $\{x_i : i\in\mathcal{X}\}$ is denoted by $(x_i)_{\forall i\in\mathcal{X}}$.
The expectation operator is denoted by $\E[\cdot]$.
For a complex number~$x$, $\lvert x \rvert$ denotes its absolute value.
For a real number~$x$, $y$, and $z$, $[x]^+ = \max(0,x)$, and $[x]_y^z = \min(\max(x,y),z)$.
For a real-valued vector~$\bx$, $[\bx]^+$ is a vector whose $i$th element is $[x_i]^+$.
We denote by $\mathbf{1}_{\{A\}}$ an indicator function taking the value of one if the statement $A$ is true, and zero otherwise.
\if\mydocumentclass2
\else
The symbol $\Leftrightarrow$ denotes the logical connective \textit{if and only if}.
The floor function is denoted by $\lfloor\cdot\rfloor$, which gives the largest integer not exceeding its argument.
\fi

\section{System Model and Problem Formulation}
\label{sec:sys}
We consider the downlink of a single cell in the MC-NOMA system, in which one single-antenna BS transmits signals to $N$~single-antenna users over $K$~subchannels.
The index sets of users and subchannels are denoted by $\mathcal{N}=\{1, 2, \ldots, N\}$ and $\mathcal{K}=\{1, 2, \ldots, K\}$, respectively.
We assume that the entire system bandwidth, $B_{\textnormal{tot}}$, is divided into $K$~orthogonal subchannels, so that there is no interference among them.
The bandwidth of Subchannel~$k$ is denoted by $B_k$.

We consider a time-slotted system over doubly block fading channels, where the channel gain of each wireless link is time-varying and frequency-selective but remains constant during a time slot and flat within a subchannel.
Let $\{h_{k,i}^t,\,t=1, 2, \ldots\}$ be the fading process associated with User~$i$ on Subchannel~$k$, where $h_{k,i}^t$ is a complex-valued continuous random variable representing the channel gain from the BS to User~$i$ on Subchannel~$k$ in time slot~$t$.
The fading process is assumed to be stationary and ergodic.
Note that the channel gain includes path loss, shadowing, and multipath fading.
We assume that information on the underlying distributions of the fading process is unknown to the BS due to the practical difficulties in obtaining such information a priori.
However, we assume that instantaneous channel gains are known to the BS at the beginning of each time slot,\footnote{This work focuses on resource scheduling from a system-level optimization perspective.
Accordingly, channel estimation is beyond the scope of this work, as in~\cite{liu2018dynamic, liu2020performance, saito2013system, parida2014power, sun2016optimal, di2016sub, lei2016power, salaun2020joint}.
For readers interested in channel estimation, we refer to~\cite{liu2014channel, mohammadian2016deterministic, sure2017survey, zheng2020intelligent} and references therein.} so that the BS can jointly perform user scheduling, subchannel assignment (i.e., user pairing per subchannel), and power allocation based on them.

In MC-NOMA, a subchannel can be assigned to multiple users simultaneously by power-domain multiplexing.
Let $x_{k,i}^t$, satisfying $\E[\lvert{x_{k,i}^t}\rvert^2]=1$, be the information-bearing signal transmitted to User~$i$ on Subchannel~$k$ in time slot~$t$, and $p_{k,i}^t$ be the power allocated for signal~$x_{k,i}^t$.
Also, let $q_{k,i}^t$ be the subchannel assignment indicator taking the value of one if Subchannel~$k$ is assigned to User~$i$ in time slot~$t$, and zero otherwise.
Then, the received signal at User~$i$ on Subchannel~$k$ in time slot~$t$ is given by
\begin{equation}\label{eq:y_i}
	y_{k,i}^t = h_{k,i}^t q_{k,i}^t \sqrt{p_{k,i}^t}\, x_{k,i}^t + \sum_{j\in\calN:j\ne i} h_{k,i}^t q_{k,j}^t \sqrt{p_{k,j}^t}\, x_{k,j}^t + n_{k,i}^t,
\end{equation}
where $n_{k,i}^t$ is the additive zero-mean complex Gaussian noise with variance~$\sigma_{k,i}^2$, and the first, second, and third terms are the desired, interference, and noise signals, respectively.
For compact notation, we define the noise-to-channel ratio (NCR) of User~$i$ on Subchannel~$k$ in time slot~$t$~as
\begin{equation}\label{eq:eta_def}
	\eta_{k,i}^t = \frac{\sigma_{k,i}^2}{\lvert h_{k,i}^t \rvert^2}.
\end{equation}
The NCR can be interpreted as the effective noise power when the channel gain is normalized to unity.

After receiving signal~$y_{k,i}^t$, User~$i$ performs SIC to decode its own signal, $x_{k,i}^t$, from it.
User~$i$ first decodes the signals for each User~$j$ whose NCR is not smaller than its NCR, i.e., $\eta_{k,j}^t \ge \eta_{k,i}^t$, and then subtracts the components associated with them from the received signal.
Then, User~$i$ decodes its own signal by treating the signals for the other users whose NCRs are smaller than its NCR as noise.
With a typical assumption that SIC has been successfully done, the maximum achievable data rate of User~$i$ on Subchannel~$k$ in time slot~$t$ is obtained as~\cite{tse2005fundamentals}
\if\mydocumentclass0
	\begin{multline}\label{eq:Rate_subchannel}
		R_{k,i}(\bp_k^t, \bq_k^t; \bh_k^t) \\
		= B_k \log_2 \left( 1 + \frac{q_{k,i}^t p_{k,i}^t}{\sum_{j\in\calN:\eta_{k,j}^t < \eta_{k,i}^t} q_{k,j}^t p_{k,j}^t + \eta_{k,i}^t} \right),
	\end{multline}
\else
	\begin{equation}\label{eq:Rate_subchannel}
		R_{k,i}(\bp_k^t, \bq_k^t; \bh_k^t)
		= B_k \log_2 \left( 1 + \frac{q_{k,i}^t p_{k,i}^t}{\sum_{j\in\calN:\eta_{k,j}^t < \eta_{k,i}^t} q_{k,j}^t p_{k,j}^t + \eta_{k,i}^t} \right),
	\end{equation}
\fi
where $\bp_k^t = (p_{k,i}^t)_{\forall i\in\calN}$, $\bq_k^t = (q_{k,i}^t)_{\forall i\in\calN}$, and $\bh_k^t = (h_{k,i}^t)_{\forall i\in\mathcal{N}}$.
From~\eqref{eq:Rate_subchannel}, the maximum achievable data rate of User~$i$ over all subchannels in time slot~$t$ is obtained as
\begin{equation}\label{eq:Rate}
	R_i(\bp^t, \bq^t; \bh^t) = \sum_{k\in\mathcal{K}} R_{k,i} (\bp_k^t, \bq_k^t; \bh_k^t),
\end{equation}
where $\bp^t = (\bp_k^t)_{\forall k\in\mathcal{K}}$, $\bq^t = (\bq_k^t)_{\forall k\in\mathcal{K}}$, and $\bh^t=(\bh_k^t)_{\forall k\in\mathcal{K}}$.
For simplicity, we interchangeably use $R_i^t$ and $R_i(\bp^t, \bq^t; \bh^t)$ without confusion.
We now define the average data rate, $\bar{R}_i$, of User~$i$ as
\begin{equation}\label{eq:avgRate}
	\bar{R}_i = \lim_{T\to\infty} \frac{1}{T} \sum_{t=1}^{T} R_i^t,
\end{equation}
and the weighted average sum rate, $\bar{R}_{\textnormal{WSR}}$, which is what we are trying to maximize, as
\begin{equation}\label{eq:WSR}
	\bar{R}_{\textnormal{WSR}} = \sum_{i\in\mathcal{N}} w_i \bar{R}_i,
\end{equation}
where $w_i$ is the weight factor representing the relative importance of User~$i$.
Additionally, each User~$i$ has its own minimum average data rate requirement, $\bar{R}_{\textnormal{min},i}$, which is represented as
\begin{equation}\label{Const:Ri}
	\bar{R}_i \ge \bar{R}_{\textnormal{min},i}, ~ \forall i\in\mathcal{N}.
\end{equation}

Because of SIC, in each time slot, the BS can schedule multiple users on the same subchannel.
However, due to the high computational complexity and the potential error propagation in SIC as well as the limited processing capabilities of users, the number of users multiplexed simultaneously on the same subchannel is typically limited to a small number, $M$.
In this regard, we define a feasible set for a subchannel assignment indicator vector, $\bq^t$, in time slot~$t$ as
\begin{equation}\label{Const:Qset}
	\mathcal{Q} = \left\{ \bq^t \in \{0,1\}^{K \times N} \suchthat \sum_{i\in\mathcal{N}} q_{k,i}^t \le M, ~ \forall k\in\mathcal{K} \right\}.
\end{equation}
By introducing this SIC capacity constraint with an appropriate $M$, we assume a perfect SIC without taking into account the error propagation of the SIC.


In addition, the BS should determine how much power to allocate to the scheduled users under given transmission power constraints.
We assume that the BS has a limited total transmission power budget of $\Pmax$ and an individual subchannel maximum power constraint of $\Pmaxk$ for each Subchannel~$k$.
In this regard, we define a feasible set for a power allocation vector, $\bp^t$, in time slot~$t$ as
\begin{equation}\label{Const:Pset}
	\mathcal{P} = \left\{ \bp^t \in \mathbb{R}^{K \times N} \suchthat
		\begin{array}{l}
			 \sum_{i\in\mathcal{N}} \sum_{k\in\mathcal{K}} p_{k,i}^t \le \Pmax,\\
			 \sum_{i\in\mathcal{N}} p_{k,i}^t \le \Pmaxk, ~ \forall k\in\mathcal{K},\\
			 p_{k,i}^t \ge 0, ~ \forall k\in\mathcal{K}, ~ \forall i\in\mathcal{N}
		\end{array}
	\right\}.
\end{equation}
Note that since the sum of the maximum powers over all subchannels is usually greater than the total transmission power budget of the BS in practice, we assume that $\sum_{k\in\mathcal{K}} \Pmaxk \ge \Pmax$.

With the performance metric function in~\eqref{eq:WSR}, the QoS constraints in~\eqref{Const:Ri}, and the feasible sets for decision variables in~\eqref{Const:Qset} and \eqref{Const:Pset},
we finally formulate the dynamic resource scheduling problem for joint user, subchannel, and power scheduling in the downlink MC-NOMA system over time-varying fading channels as
\begin{IEEEeqnarray}{c'c'l}
	\IEEEnonumber*
	\Pone & \maximize_{\bp^t,\,\bq^t,\,\forall t} & \bar{R}_{\textnormal{WSR}} \label{Prob:Obj}\\
	&\subjto & \bar{R}_i \ge \bar{R}_{\textnormal{min},i}, ~ \forall i\in\mathcal{N},\\
	&& \bp^t \in \mathcal{P}, ~ \bq^t \in \mathcal{Q}, ~ \forall t.
\end{IEEEeqnarray}
We first note that dealing with Problem~$\Pone$ is not easy due to the nonconcave objective function, the QoS constraints composed of neither convex nor concave functions, integer decision variables, and the average operation over an infinite time horizon.
To resolve these challenges, we first develop an opportunistic MC-NOMA scheduling algorithm by leveraging the Lagrangian duality and the stochastic optimization theory, whereby Problem~$\Pone$ is decomposed into a series of deterministic optimization subproblems, one for each time slot.
More specifically, a subproblem is the Joint-SAPA problem to maximize the instantaneous weighted sum rate without the QoS constraints in that time slot, which will be defined as Problem~$\Pt$ in the next section.
As a consequence, we no longer need to solve Problem~$\Pone$ directly at once, but rather solve the Joint-SAPA problem at each time slot in an online manner without considering the average operation over an infinite time horizon and the QoS constraints.
Meanwhile, an important caveat is that the Joint-SAPA problem needs to be solved by a simple algorithm with very low computational complexity so that the BS can generate and transmit signals in every short time slot.
Hence, we develop a heuristic algorithm to solve the Joint-SAPA problem with low computational complexity, called Joint-SAPA-LCC algorithm.
The flow chart of the process for solving the dynamic resource scheduling problem, Problem~$\Pone$, is schematically illustrated in Fig.~\ref{fig:flow chart}.
Note that the Joint-SAPA-LCC algorithm is a built-in algorithm that runs every time slot within the opportunistic MC-NOMA scheduling algorithm.
In the following, for ease of explanation, we first develop the Joint-SAPA-LCC algorithm in Section~\ref{sec:Joint-SAPA} and then the opportunistic MC-NOMA scheduling algorithm in Section~\ref{sec:scheduling}.

\begin{figure}[!t]
\centering
\includegraphics[width=12cm]{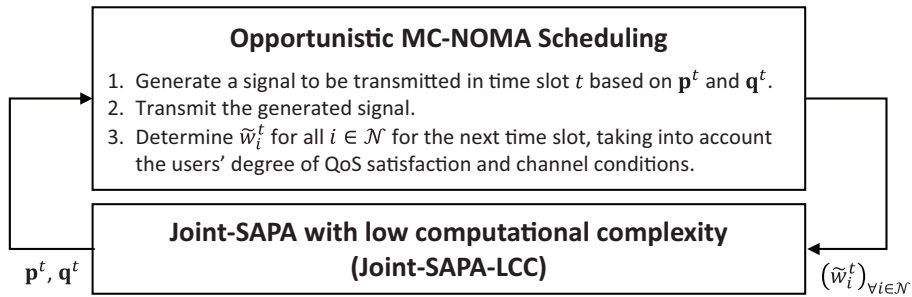}
\caption{The flow chart of the opportunistic MC-NOMA scheduling algorithm operating in an online manner to solve Problem~$\Pone$.}
\label{fig:flow chart}
\end{figure}

\section{Joint-SAPA with low computational complexity (Joint-SAPA-LCC)}
\label{sec:Joint-SAPA}
In this section, we develop our Joint-SAPA-LCC algorithm that solves the instantaneous weighted sum rate maximization problem for each time slot~$t$, defined by
\begin{IEEEeqnarray}{c'c'l}\label{Prob:Dual_s2}
	\IEEEnonumber*
	\Pt & \maximize_{\bp^t\in\mathcal{P},\,\bq^t\in\mathcal{Q}} & \sum_{i\in\calN} \tilde{w}_i^t R_i^t,
\end{IEEEeqnarray}
where $\tilde{w}_i^t$ is the effective weight of User~$i$ in time slot $t$.
As mentioned before, this problem is not subject to the QoS constraints.
Instead, the effective weights are systematically adjusted every time slot so that the QoS constraints in~\eqref{Const:Ri} are met.
It is worth noting that the effective weight, $\tilde{w}_i^t$, differs from the weight, $w_i$, in~\eqref{eq:WSR} in that it is systematically adjusted in every time slot based on the degree of QoS satisfaction of User~$i$ and its channel condition.
This systematic update process will be rigorously elaborated later in Section~\ref{sec:scheduling}.
Although the QoS constraints are not considered in Problem~$\Pt$, it is still an NP-hard problem and very difficult to solve using conventional methods since it contains not only a nonconcave objective function but also integer variables.
Furthermore, we need a fast solution because the transmission signal should be generated and transmitted according to the solution in every short time slot.
For these reasons, we develop a heuristic suboptimal algorithm that provides a near-optimal performance despite very low computational complexity.
In the remainder of this section, since Problem~$\Pt$ is focusing only on time slot~$t$, we omit the superscript~$t$ for notational brevity.

To solve Problem~$\Pt$, we exploit the primal decomposition method~\cite{boyd2007notes,palomar2006tutorial}.
By introducing a new coupling vector $\bar{\bP}=(\bar{P}_k)_{\forall k\in\mathcal{K}}$, we reformulate Problem~$\Pt$ equivalently as
\begin{IEEEeqnarray}{c'c'l}
	\IEEEnonumber*
	\Pteq & \maximize_{\bp,\,\bq,\,\bar{\bP}} & \sum_{i\in\calN} \tilde{w}_i \sum_{k\in\mathcal{K}} R_{k,i}(\bp_k, \bq_k; \bh_k) \\
	&\subjto
	&\sum_{k\in\mathcal{K}} \bar{P}_k \le \Pmax, \label{Prob:Done:Const1}\\
	&&0 \le \bar{P}_k \le \Pmaxk, ~ \forall k\in\mathcal{K}, \label{Prob:Done:Const2}\\
	&&\sum_{i\in\mathcal{N}} p_{k,i} \le \bar{P}_k, ~ \forall k\in\mathcal{K}, \label{Prob:Done:Const3}\\
	&&p_{k,i} \ge 0, ~ \forall k\in\mathcal{K}, ~ \forall i\in\mathcal{N}, \label{Prob:Done:Const4}\\
	&&\sum_{i\in\mathcal{N}} q_{k,i} \le M, ~ \forall k\in\mathcal{K}, \label{Prob:Done:Const5}\\
	&&q_{k,i} \in \{0,1\}, ~ \forall k\in\mathcal{K}, ~ \forall i\in\mathcal{N}. \label{Prob:Done:Const6}
\end{IEEEeqnarray}
It is worth noting that Problem~$\Pteq$ could be decoupled for each Subchannel~$k$ if the coupling vector~$\bar{\bP}$ were fixed.
Hence, we separate it into two levels of optimization.
At the lower level, we have $K$~subproblems, one for each Subchannel~$k$, defined by
\begin{IEEEeqnarray}{c'c'l}
	\IEEEnonumber*
	\Slave & \maximize_{\bp_k,\,\bq_k} & \sum_{i\in\calN} \tilde{w}_i R_{k,i}(\bp_k, \bq_k; \bh_k) \\
	&\subjto
	&\sum_{i\in\mathcal{N}} p_{k,i} \le \bar{P}_k, \\
	&&p_{k,i} \ge 0, ~ \forall i\in\mathcal{N}, \\
	&&\sum_{i\in\mathcal{N}} q_{k,i} \le M, \\
	&&q_{k,i} \in \{0,1\}, ~ \forall i\in\mathcal{N}.
\end{IEEEeqnarray}
At the higher level, we have an optimization in charge of updating the coupling vector $\bar{\bP}$, defined by
\begin{IEEEeqnarray}{c'c'l}
	\IEEEnonumber*
	\Master & \maximize_{\bar{\bP}} & \phi^*(\bar{\bP}) = \sum_{k\in\mathcal{K}} \phi_k^*(\bar{P}_k) \\
	&\subjto
	&\sum_{k\in\mathcal{K}} \bar{P}_k \le \Pmax, \\
	&&0 \le \bar{P}_k \le \Pmaxk, ~ \forall k\in\mathcal{K}, 
\end{IEEEeqnarray}
where $\phi_k^*(\bar{P}_k)$ is the optimal value of Problem~$\Slave$.
Then, we can obtain a suboptimal solution, $\{\bp^*, \bq^*, \bar{\bP}^*\}$, to Problem~$\Pteq$ by alternately solving Problems~$\Slave$, $\forall k\in\mathcal{K}$, and Problem~$\Master$ until convergence.
The pseudocode for this process is summarized in Algorithm~\ref{Alg:SG-Master}.
We can show that Algorithm~\ref{Alg:SG-Master} converges to a stationary point.

\begin{theorem} \label{thm:convergence}
Algorithm~\ref{Alg:SG-Master} converges to a stationary point.
\end{theorem}

\begin{proof}
\if\mydocumentclass2
Due to the page limit, we leave the proof in Appendix~A of the arXiv version~\cite{fullVer}.
\else
See Appendix~\ref{prove:thm:convergence}.
\fi
\end{proof}

In the primal decomposition method, if the primal problem is a convex problem, not only the subproblems but also the master problem becomes a convex problem, resulting in the convergence to a global optimal solution.
In our case, however, since Problem~$\Pt$ is not a convex problem, the convergence to a global optimal solution is not guaranteed.
Nevertheless, we show by simulation results that the algorithm provides a near-optimal performance with very low computational complexity.\footnote{In~\cite{salaun2020joint}, it is demonstrated that the proposed Joint-SAPA algorithm therein (called Joint-SAPA-DP later in simulation results) achieves near-optimal performance.
Accordingly, by comparing our Joint-SAPA-LCC algorithm with the Joint-SAPA-DP algorithm via simulation, we verify that our algorithm provides near-optimal performance with very low computational complexity.}
Also, it is worth noting that we will derive $\phi_k^*(\bar{P}_k)$ in a closed form in Section~\ref{sec:Solution_slave}, and effectively use it to solve Problem~$\Master$ in Section~\ref{sec:Solution_master}.

\begin{algorithm}[!t]
\DontPrintSemicolon
Initialize: $\bar{P}_k \gets \Pmax/K, \forall k\in\mathcal{K}$.\\
\Repeat{convergence}{
	Obtain $\phi_k^*(\bar{P}_k)$ by solving Problem~$\Slave$ for the given $\bar{P}_k$, $\forall k\in\mathcal{K}$, using Algorithm~\ref{Alg:Solve_Slave}.\\
	Obtain the solution, $\bar{\bP}^*$, to Problem~$\Master$ for the given $\phi^*(\bar{\bP})$ using Algorithm~\ref{Alg:Solve_Master}.\\
	Update $\bar{\bP} \gets \bar{\bP}^*$
%
}
Obtain $\{\bp_k^*, \bq_k^*\}$ by solving Problem~$\Slave$ for the given $\bar{P}_k$, $\forall k\in\mathcal{K}$, using Algorithm~\ref{Alg:Solve_Slave}.\\
\Return{$\{\bp^*, \bq^*\}$.}
\caption{Joint-SAPA-LCC}
\label{Alg:SG-Master}
\end{algorithm}

\subsection{Solution to Problem $\Slave$}
\label{sec:Solution_slave}
In this subsection, we discuss how to solve Problem~$\Slave$ for a given $\bar{P}_k$.
To be specific, to maximize the weighted sum rate over Subchannel~$k$, we find out which users to be assigned to Subchannel~$k$ and how much power to be allocated to them under the limited power of $\bar{P}_k$.
In addition, we derive the corresponding objective value, $\phi_k^*(\bar{P}_k)$, as a closed-form function of $\bar{P}_k$.

Problem~$\Slave$ is difficult to solve mainly due to the integer variables, i.e., the subchannel assignment indicators, $\bq_k$.
To address this difficulty, we first consider the following problem:
\begin{IEEEeqnarray}{c'c'l}
	\IEEEnonumber*
	\Q & \maximize_{\bp_k} & \sum_{i\in\calN} \tilde{w}_i R_{k,i}(\bp_k; \bh_k) \\
	&\subjto
	&\sum_{i\in\mathcal{N}} p_{k,i} \le \bar{P}_k, \\
	&&p_{k,i} \ge 0, ~ \forall i\in\mathcal{N},
\end{IEEEeqnarray}
where
\begin{equation}
	R_{k,i}(\bp_k; \bh_k) = B_k \log_2 \left( 1 + \frac{p_{k,i}}{\sum_{j\in\calN : \eta_{k,j} < \eta_{k,i}} p_{k,j} + \eta_{k,i}} \right).
\end{equation}
Although Problem~$\Q$ is different from Problem~$\Slave$, we can easily derive the optimal solution to Problem~$\Slave$ from that to Problem~$\Q$ if a certain condition is met, by the following~theorem.

\begin{theorem}\label{thm:user_selection_exclusion}
	Let $\bp_k^\dagger=(p_{k,i}^\dagger)_{\forall i\in\mathcal{N}}$ be an optimal solution to Problem~$\Q$, and suppose that it satisfies
	\begin{equation}\label{eq:condition_le_M}
		\sum_{i\in\mathcal{N}} \mathbf{1}_{\left\{p_{k,i}^\dagger>0\right\}} \le M.
	\end{equation}
	Then, an optimal solution, $\{\bp_k^*, \bq_k^*\}$, to Problem~$\Slave$ can be obtained as
	\begin{equation}\label{eq:pstar qstar from pdagger}
		\bp_k^* = \bp_k^\dagger \textnormal{ and } \bq_k^* = (q_{k,i}^*)_{\forall i\in\mathcal{N}},
	\end{equation}
	where $q_{k,i}^* = \mathbf{1}_{\left\{p_{k,i}^\dagger>0\right\}}$ for all $i\in\mathcal{N}$.
%
\end{theorem}

\begin{proof}
\if\mydocumentclass2
Due to the page limit, we leave the proof in Appendix~B of the arXiv version~\cite{fullVer}.
\else
See Appendix~\ref{prove:thm:user_selection_exclusion}.
\fi
\end{proof}

Note that in the remaining of this subsection, we focus on developing a low-complexity heuristic algorithm that solves Problem~$\Q$ without considering \eqref{eq:condition_le_M}.
However, it will be revealed later that the solution to Problem~$\Q$ obtained by our proposed algorithm always satisfies \eqref{eq:condition_le_M} as long as $M>1$.
Hence, a solution to Problem~$\Slave$ can be easily obtained from the solution to Problem~$\Q$ based on Theorem~\ref{thm:user_selection_exclusion}.


We now discuss how to solve Problem~$\Q$.
Even though Problem~$\Q$ does not have any integer variables, it is still known to be strongly NP-hard~\cite{salaun2018optimal}.
Hence, we focus on developing a heuristic algorithm that provides a near-optimal solution to Problem~$\Q$ with very low computational complexity.
To this end, we first find \textit{candidate} users who might be allocated positive power on Subchannel~$k$.
Then, we derive the optimal power allocation for them in closed forms.
For compact notation, we assume, without loss of generality, that users are ordered such that $\eta_{k,i} > \eta_{k,j}$ if $i < j$, and define a \textit{last SIC user} as follows.

\begin{definition}
\textit{A last SIC user} refers to a user who does not experience any interference signals after the SIC process.
\end{definition}

We start with the assumption that User~$\varphi_k$ is the last SIC user on Subchannel~$k$ and has been allocated a certain amount of power.
Accordingly, we assume that $p_{k,\varphi_k}$ is given as a fixed positive value, and $p_{k,i}$ for $i>\varphi_k$ is given as zero so that User~$\varphi_k$ does not experience any interference signals after the SIC process.
%
Under this assumption, $p_{k,i}$'s for $i\ge\varphi_k$ are no longer decision variables.
Note that how to select the last SIC user and how much power to allocate to it will be discussed later.
Now, Problem~$\Q$ can be reformulated as
\if\mydocumentclass0
	\begin{IEEEeqnarray}{c'c'l}
		\IEEEnonumber*
		\Qone & \maximize_{p_{k,i},\,i<\varphi_k}
			&\begin{multlined}[t]
				\sum_{i=1}^{\varphi_k-1} \tilde{w}_i B_k \log_2 \left(1+\frac{p_{k,i}}{\sum_{j>i} p_{k,j} + \eta_{k,i}} \right) \\
				+ \tilde{w}_{\varphi_k} B_k \log_2 \left( 1+ \frac{p_{k,\varphi_k}}{\eta_{k,\varphi_k}} \right)
			\end{multlined} \\
		&\subjto
			&\sum_{i=1}^{\varphi_k-1} p_{k,i} + p_{k,\varphi_k} \le \bar{P}_k, \\
		&	&p_{k,i} \ge 0, ~ i<\varphi_k.
	\end{IEEEeqnarray}
\else
	\begin{IEEEeqnarray}{c'c'l}
		\IEEEnonumber*
		\Qone & \maximize_{p_{k,i},\,i<\varphi_k}
			& \sum_{i=1}^{\varphi_k-1} \tilde{w}_i B_k \log_2 \left(1+\frac{p_{k,i}}{\sum_{j>i} p_{k,j} + \eta_{k,i}} \right) + \tilde{w}_{\varphi_k} B_k \log_2 \left( 1+ \frac{p_{k,\varphi_k}}{\eta_{k,\varphi_k}} \right)\\
		&\subjto
			&\sum_{i=1}^{\varphi_k-1} p_{k,i} + p_{k,\varphi_k} \le \bar{P}_k, \\
		&	&p_{k,i} \ge 0, ~ i<\varphi_k.
	\end{IEEEeqnarray}
\fi
The purpose of this problem is not to find power allocation to users but to find \textit{candidate} users when User~$\varphi_k$ is selected as the last SIC user on Subchannel~$k$.
Note that, in general, enough power is usually allocated to the last SIC user's signal for its successful decoding.
Accordingly, the amount of interference power experienced by users other than the last SIC user will usually be non-negligible and exceed the thermal noise power.
In other words, we assume that the mode of operation is interference limited, so that the noise power of the users suffering from the interference signals is neglected, i.e., $\sigma_{k,i}^2=0$ (accordingly, $\eta_{k,i}=0$) for $i<\varphi_k$.
Then, letting $\rho_{k,i} = \sum_{j=i}^{\varphi_k} p_{k,j}$ for $i\le\varphi_k$, we can approximate Problem~$\Qone$ as
%
%
\if\mydocumentclass0
	\begin{IEEEeqnarray}{c'c'l}
		\IEEEnonumber*
		\Qtwo & \maximize_{\rho_{k,i},\,i<\varphi_k}
			&\begin{multlined}[t]
				\sum_{i=1}^{\varphi_k-1} \tilde{w}_i B_k \log_2 \left( \frac{\rho_{k,i}}{\rho_{k,i+1}} \right) \\
				+ \tilde{w}_{\varphi_k} B_k \log_2 \left( 1+ \frac{\rho_{k,\varphi_k}}{\eta_k,\varphi_k} \right)
			\end{multlined} \\
		&\subjto
			&\prod_{i=1}^{\varphi_k-1}\frac{\rho_{k,i}}{\rho_{k,i+1}} \le \frac{\bar{P}_k}{\rho_{k,\varphi_k}}, \\
		&	&\frac{\rho_{k,i}}{\rho_{k,i+1}} \ge 1, ~ i<\varphi_k,
	\end{IEEEeqnarray}
\else
	\begin{IEEEeqnarray}{c'c'l}
	\IEEEnonumber*
		\Qtwo & \maximize_{\rho_{k,i},\,i<\varphi_k}
			& \sum_{i=1}^{\varphi_k-1} \tilde{w}_i B_k \log_2 \left( \frac{\rho_{k,i}}{\rho_{k,i+1}} \right) + \tilde{w}_{\varphi_k} B_k \log_2 \left( 1+ \frac{\rho_{k,\varphi_k}}{\eta_k,\varphi_k} \right)\\
		&\subjto
			&\prod_{i=1}^{\varphi_k-1}\frac{\rho_{k,i}}{\rho_{k,i+1}} \le \frac{\bar{P}_k}{\rho_{k,\varphi_k}}, \\
		&	&\frac{\rho_{k,i}}{\rho_{k,i+1}} \ge 1, ~ i<\varphi_k,
	\end{IEEEeqnarray}
\fi
where the inequality constraints are equivalent to those in Problem~$\Qone$, which can be derived by simple arithmetic operations.
In succession, by letting $r_{k,i} = \log_2(\rho_{k,i}/\rho_{k,i+1})$ for $i<\varphi_k$ and $r_{k,\varphi_k} = \log_2(\rho_{k,\varphi_k})$, and taking the logarithm of the both sides of the constraints, we can reformulate Problem~$\Qtwo$ equivalently as
\begin{IEEEeqnarray}{c'c'l}
	\IEEEnonumber*
	\Qthree & \maximize_{r_{k,i},\,i<\varphi_k}
		&\sum_{i=1}^{\varphi_k-1} \tilde{w}_i B_k r_{k,i} + \tilde{w}_{\varphi_k} B_k \log_2 \left( 1+ \frac{2^{r_{k,\varphi_k}}}{\eta_k,\varphi_k} \right) \\
	&\subjto
		&\sum_{i=1}^{\varphi_k-1} r_{k,i} \le \log_2(\bar{P}_k) - r_{k,\varphi_k}, \\
	&	&r_{k,i} \ge 0, ~ i<\varphi_k.
\end{IEEEeqnarray}
Note that since $r_{k,\varphi_k}$ is not a decision variable, the second term, $\tilde{w}_{\varphi_k} B_k \log_2 (1+{2^{r_{k,\varphi_k}}}/{\eta_k,\varphi_k})$, of the objective function and the right-hand side, $\log_2(\bar{P}_k)-r_{k,\varphi_k}$, of the first constraint are constants.
Also, the decision variables are linearly combined in the objective function, and the feasible set is a unit simplex.
Hence, it is obvious that the objective function is maximized when all the decision variables, except the one with the largest weight, are zero.
Also, by the definition of $r_{k,i}$, we can easily see that, for any $i<\varphi_k$, $p_{k,i}$ is zero if and only if $r_{k,i}$ is zero.
Thus, only one user with the largest weight is selected as the other \textit{candidate} user on Subchannel~$k$ together with the last SIC user, i.e., User~$\varphi_k$.
We state this result in the following theorem.
\begin{theorem}\label{thm:two user selected}
	Under the assumption that the noise power of users suffering from the interference signals is neglected, on each Subchannel~$k$, at most two users are allocated power according to the optimal solution to Problem~$\Qone$.
	To be specific, when User~$\varphi_k$ has been selected as the last SIC user on Subchannel~$k$, User~$\psi_k(\varphi_k)$ is accordingly selected as the other \textit{candidate} user, where
	\begin{equation}\label{eq:psi_k}
		\psi_k(\varphi_k) = \argmax_{i<\varphi_k} \{\tilde{w}_i\}.
	\end{equation}
\end{theorem}

\noindent By Theorem~\ref{thm:two user selected}, we can reduce Problem~$\Qone$ to the power allocation problem for the two-user case as
\if\mydocumentclass0
	\begin{IEEEeqnarray}{c'c'l}
		\IEEEnonumber*
		\Qfour & \maximize_{p_{k,\psi_k},\,p_{k,\varphi_k}}
			&\begin{multlined}[t]
				\tilde{w}_{\psi_k} B_k \log_2 \left( 1 + \frac{p_{k,\psi_k}}{p_{k,\varphi_k} + \eta_{k,\psi_k}} \right) \\
				+ \tilde{w}_{\varphi_k} B_k \log_2 \left( 1 + \frac{p_{k,\varphi_k}}{\eta_{k,\varphi_k}} \right)
			\end{multlined} \\
		&\subjto
			&p_{k,\psi_k} + p_{k,\varphi_k} \le \bar{P}_k, \\
		&	&p_{k,\psi_k} \ge 0, ~ p_{k,\varphi_k} \ge 0,
	\end{IEEEeqnarray}
\else
	\begin{IEEEeqnarray}{c'c'l}
		\IEEEnonumber*
		\Qfour & \maximize_{p_{k,\psi_k},\,p_{k,\varphi_k}}
			&\tilde{w}_{\psi_k} B_k \log_2 \left( 1 + \frac{p_{k,\psi_k}}{p_{k,\varphi_k} + \eta_{k,\psi_k}} \right) + \tilde{w}_{\varphi_k} B_k \log_2 \left( 1 + \frac{p_{k,\varphi_k}}{\eta_{k,\varphi_k}} \right) \\
		&\subjto
			&p_{k,\psi_k} + p_{k,\varphi_k} \le \bar{P}_k, \\
		&	&p_{k,\psi_k} \ge 0, ~ p_{k,\varphi_k} \ge 0,
	\end{IEEEeqnarray}
\fi
where $\psi_k(\varphi_k)$ is replaced with $\psi_k$ for notational simplicity.
This two-user power allocation problem can be optimally solved in closed forms.

\begin{theorem}\label{thm:2-user}
Let $\{p_{k,\psi_k}^\star, p_{k,\varphi_k}^\star\}$ be the optimal solution to Problem~$\Qfour$.
It can be obtained as
\begin{align}
	&p_{k,\varphi_k}^\star =
		\begin{dcases}
			0,			& \textnormal{if } \tilde{w}_{\varphi_k} / \tilde{w}_{\psi_k} \le C_k^1, \\
			\bar{P}_k,	& \textnormal{if } \tilde{w}_{\varphi_k} / \tilde{w}_{\psi_k} > C_k^2(\bar{P}_k), \\
			\frac{\tilde{w}_{\psi_k} \eta_{k,\varphi_k} - \tilde{w}_{\varphi_k} \eta_{k,\psi_k}}{\tilde{w}_{\varphi_k} - \tilde{w}_{\psi_k}},
						& \textnormal{otherwise,}
		\end{dcases} \label{eq:optimal p_SIC}\\
	&p_{k,\psi_k}^\star = \bar{P}_k - p_{k,\varphi_k}^\star, \label{eq:optimal p_nonSIC}
\end{align}
where
\begin{equation}\label{eq:C1 and C2}
	C_k^1 = \frac{\eta_{k,\varphi_k}}{\eta_{k,\psi_k}} \quad \textnormal{and} \quad
	C_k^2(\bar{P}_k) = \frac{\bar{P}_k + \eta_{k,\varphi_k}}{\bar{P}_k + \eta_{k,\psi_k}}.
\end{equation}
\end{theorem}

\begin{proof}
\if\mydocumentclass2
Due to the page limit, we leave the proof in Appendix~C of the arXiv version~\cite{fullVer}.
\else
See Appendix~\ref{prove:thm:2-user}.
\fi
\end{proof}

All derivations so far have been made on the assumption that User~$\varphi_k$ is allocated a certain amount of power as the last SIC user.
However, in the case where $\tilde{w}_{\varphi_k}/\tilde{w}_{\psi_k(\varphi_k)} \le C_k^1$, the optimal power allocation to the last SIC user becomes zero.
Thus, we consider this case to be contradictory, and set the objective value (i.e., the weighted sum rate on Subchannel~$k$) to a negative infinity in this case.
Then, the weighted sum rate on Subchannel~$k$ under the assumption that User~$\varphi_k$ is the last SIC user on this subchannel can be obtained as
\if\mydocumentclass0
	\begin{align}\label{eq:RSIC}
		&\phi_k(\bar{P}_k ; \varphi_k) \nonumber\\
		&=\begin{dcases}
			\begin{multlined}[b]
				\tilde{w}_{\psi_k} B_k \log_2 \left( 1 + \frac{p_{k,\psi_k}^\star}{p_{k,\varphi_k}^\star + \eta_{k,\psi_k}} \right) \\
				+ \tilde{w}_{\varphi_k} B_k \log_2 \left( 1 + \frac{p_{k,\varphi_k}^\star}{\eta_{k,\varphi_k}} \right),
			\end{multlined}
				& \textnormal{if } \tilde{w}_{\varphi_k}/\tilde{w}_{\psi_k} > C_k^1, \\
			-\infty,
				& \textnormal{otherwise,}
		\end{dcases}
	\end{align}
\else
	\begin{equation}\label{eq:RSIC}
		\phi_k(\bar{P}_k ; \varphi_k) =
		\begin{dcases}
			\tilde{w}_{\psi_k} B_k \log_2 \left( 1 + \frac{p_{k,\psi_k}^\star}{p_{k,\varphi_k}^\star + \eta_{k,\psi_k}} \right) + \tilde{w}_{\varphi_k} B_k \log_2 \left( 1 + \frac{p_{k,\varphi_k}^\star}{\eta_{k,\varphi_k}} \right), & \textnormal{if } \tilde{w}_{\varphi_k}/\tilde{w}_{\psi_k} > C_k^1,\\
			-\infty, & \textnormal{otherwise,}
		\end{dcases}
	\end{equation}
\fi
where $p_{k,\varphi_k}^\star$ and $p_{k,\psi_k}^\star$ are given by \eqref{eq:optimal p_SIC} and \eqref{eq:optimal p_nonSIC}, respectively.
Then, the optimal last SIC user on Subchannel~$k$, indexed by $\varphi_k^*$, can be determined as
\begin{equation}\label{eq:Lstar}
	\varphi_k^* = \argmax_{\varphi_k\in\mathcal{N}} \phi_k (\bar{P}_k ; \varphi_k),
\end{equation}
and the corresponding optimal value is given by
\begin{equation}\label{eq:phik}
	\phi_k^*(\bar{P}_k) = \phi_k (\bar{P}_k ; \varphi_k^*).
\end{equation}
Consequently, Users~$\varphi_k^*$ and $\psi_k(\varphi_k^*)$ are selected as the optimal \textit{candidate} users on Subchannel~$k$.
Then, by Theorem~\ref{thm:2-user}, the power allocation solution to Problem~$\Q$ can be obtained as $\bp_k^\dagger=(p_{k,i}^\dagger)_{\forall i\in\mathcal{N}}$, where $p_{k,\varphi_k^*}^\dagger = p_{k,\varphi_k^*}^\star$, $p_{k,\psi_k(\varphi_k^*)}^\dagger = p_{k,\psi_k(\varphi_k^*)}^\star$, and $p_{k,i}^\dagger=0$ for all $i\in\mathcal{N}\setminus\{\varphi_k^*,\psi_k(\varphi_k^*)\}$.
Finally, the solution, $\{\bp_k^*, \bq_k^*\}$, to Problem~$\Slave$ can be derived from $\bp_k^\dagger$ using Theorem~\ref{thm:user_selection_exclusion}.
We summarize this process in Algorithm~\ref{Alg:Solve_Slave}.

\begin{algorithm}[!t]
\DontPrintSemicolon
\For{each User~$\varphi_k\in\calN$}{
	Suppose that User~$\varphi_k$ is the last SIC user.\\
	Select the other \textit{candidate} user using \eqref{eq:psi_k}.\\
	Obtain their power allocation using Theorem~\ref{thm:2-user}.\\
	Obtain $\phi_k(\bar{P}_k ; \varphi_k)$ using~\eqref{eq:RSIC}.\\
}
Obtain~$\varphi_k^*$ and $\phi_k^*(\bar{P}_k)$ using~\eqref{eq:Lstar} and \eqref{eq:phik}, respectively.\\
Let Users~$\varphi_k^*$ and $\psi_k(\varphi_k^*)$ be the optimal \textit{candidate} users.\\
Obtain $\bp_k^\dagger$ using Theorem~\ref{thm:2-user}.\\
Derive $\{\bp_k^*,\bq_k^*\}$ from $\bp_k^\dagger$ using Theorem~\ref{thm:user_selection_exclusion}.\\
\Return{$\{\bp_k^*, \bq_k^*, \phi_k^*(\bar{P}_k)\}$.}
\caption{Solution to Problem~$\Slave$}
\label{Alg:Solve_Slave}
\end{algorithm}

Before moving on to the next subsection, we analyze the computational complexity of Algorithm~\ref{Alg:Solve_Slave}.
First, once any one user is selected as the last SIC user, the computational complexity to find the other \textit{candidate} user based on \eqref{eq:psi_k} is $\mathcal{O}(N)$.
Then, thanks to the closed-form power allocation formulas in Theorem~\ref{thm:2-user}, the two \textit{candidate} users' power allocation and the corresponding weighted sum rate can be calculated in $\mathcal{O}(1)$.
Consequently, since the number of cases in which any one user is selected as the last SIC user is at most $N$, the overall computational complexity of Algorithm~\ref{Alg:Solve_Slave} is $\mathcal{O}(N^2)$.

\subsection{Solution to Problem~$\Master$}
\label{sec:Solution_master}
In this subsection, we discuss how to solve Problem~$\Master$ for the given $\phi^*(\bar{\bP}) = \sum_{k\in\mathcal{K}} \phi_k^*(\bar{P}_k)$.
To be specific, we find an optimal coupling vector, $\bar{\bP}^*$, that maximizes $\phi^*(\bar{\bP})$ under the limited total transmission power budget of $\Pmax$.
For compact notation, we write the optimal last SIC user, $\varphi_k^*$, and the other \textit{candidate} user, $\psi_k(\varphi_k^*)$, on Subchannel~$k$ as $\varphi_k$ and $\psi_k$, respectively, if there is no confusion.

By plugging the optimal power allocation solution derived in Theorem~\ref{thm:2-user} into \eqref{eq:phik}, we can obtain $\phi_k^*(\bar{P}_k)$ as
\if\mydocumentclass0
	\begin{align}\label{eq:phik2}
		&\phi_k^*(\bar{P}_k) \nonumber\\
		&=\begin{dcases}
			\tilde{w}_{\psi_k} B_k \log_2 \left( 1 + \frac{\bar{P}_k}{\eta_{k,\psi_k}} \right) + C_k^3, & \textnormal{if } \tilde{w}_{\varphi_k} / \tilde{w}_{\psi_k} \le C_k^2(\bar{P}_k), \\
			\tilde{w}_{\varphi_k} B_k \log_2 \left( 1 + \frac{\bar{P}_k}{\eta_{k,\varphi_k}} \right), & \textnormal{otherwise},
		\end{dcases}
	\end{align}
\else
	\begin{equation}\label{eq:phik2}
		\phi_k^*(\bar{P}_k) =
		\begin{dcases}
			\tilde{w}_{\psi_k} B_k \log_2 \left( 1 + \frac{\bar{P}_k}{\eta_{k,\psi_k}} \right) + C_k^3, & \textnormal{if } \tilde{w}_{\varphi_k} / \tilde{w}_{\psi_k} \le C_k^2(\bar{P}_k), \\
			\tilde{w}_{\varphi_k} B_k \log_2 \left( 1 + \frac{\bar{P}_k}{\eta_{k,\varphi_k}} \right), & \textnormal{otherwise},
		\end{dcases}
	\end{equation}
\fi
where $C_k^2(\bar{P}_k)$ is defined in~\eqref{eq:C1 and C2}, and
\if\mydocumentclass0
	\begin{multline} \label{eq:C3}
		C_k^3 = \tilde{w}_{\psi_k} B_k \log_2 \left( \frac{\tilde{w}_{\varphi_k} - \tilde{w}_{\psi_k}}{\eta_{k,\varphi_k}-\eta_{k,\psi_k}} \cdot \frac{\eta_{k,\psi_k}}{\tilde{w}_{\psi_k}} \right) \\
		+ \tilde{w}_{\varphi_k} B_k \log_2 \left( \frac{\eta_{k,\varphi_k}-\eta_{k,\psi_k}}{\tilde{w}_{\varphi_k}-\tilde{w}_{\psi_k}} \cdot \frac{\tilde{w}_{\varphi_k}}{\eta_{k,\varphi_k}} \right).
	\end{multline}
\else
	\begin{equation} \label{eq:C3}
		C_k^3 = \tilde{w}_{\psi_k} B_k \log_2 \left( \frac{\tilde{w}_{\varphi_k} - \tilde{w}_{\psi_k}}{\eta_{k,\varphi_k}-\eta_{k,\psi_k}} \cdot \frac{\eta_{k,\psi_k}}{\tilde{w}_{\psi_k}} \right)
		+ \tilde{w}_{\varphi_k} B_k \log_2 \left( \frac{\eta_{k,\varphi_k}-\eta_{k,\psi_k}}{\tilde{w}_{\varphi_k}-\tilde{w}_{\psi_k}} \cdot \frac{\tilde{w}_{\varphi_k}}{\eta_{k,\varphi_k}} \right).
	\end{equation}
\fi
Note that, in~\eqref{eq:phik2}, the case where $\tilde{w}_{\varphi_k} / \tilde{w}_{\psi_k} \le C_k^1$ is excluded since $\tilde{w}_{\varphi_k} / \tilde{w}_{\psi_k}$ is always greater than $C_k^1$ as long as the optimal \textit{candidate} users are selected by Algorithm~\ref{Alg:Solve_Slave}.
In succession, since dealing with~\eqref{eq:phik2} is difficult due to $C_k^2(\bar{P}_k)$ that varies with $\bar{P}_k$, we reformulate it equivalently as
\if\mydocumentclass0
\begin{align}\label{eq:phik3}
	&\phi_k^*(\bar{P}_k) \nonumber\\
	&=\begin{dcases}
		\tilde{w}_{\psi_k} B_k \log_2 \left( 1 + \frac{\bar{P}_k}{\eta_{k,\psi_k}} \right) + C_k^3, & \begin{multlined}[t]\textnormal{if } \tilde{w}_{\varphi_k}/\tilde{w}_{\psi_k} < 1 \\\textnormal{ and } \bar{P}_k \ge C_k^4, \end{multlined}\\
		\tilde{w}_{\varphi_k} B_k \log_2 \left( 1 + \frac{\bar{P}_k}{\eta_{k,\varphi_k}} \right), & \textnormal{otherwise,}
	\end{dcases}
\end{align}
\else
\begin{equation}\label{eq:phik3}
	\phi_k^*(\bar{P}_k) = 
	\begin{dcases}
		\tilde{w}_{\psi_k} B_k \log_2 \left( 1 + \frac{\bar{P}_k}{\eta_{k,\psi_k}} \right) + C_k^3, & \textnormal{if } \tilde{w}_{\varphi_k}/\tilde{w}_{\psi_k} < 1 \textnormal{ and } \bar{P}_k \ge C_k^4,\\
		\tilde{w}_{\varphi_k} B_k \log_2 \left( 1 + \frac{\bar{P}_k}{\eta_{k,\varphi_k}} \right), & \textnormal{otherwise,}
	\end{dcases}
\end{equation}
\fi
where
\begin{equation} \label{eq:C4}
	C_k^4 = \frac{\tilde{w}_{\psi_k} \eta_{k,\varphi_k} - \tilde{w}_{\varphi_k} \eta_{k,\psi_k}}{\tilde{w}_{\varphi_k}-\tilde{w}_{\psi_k}}.
\end{equation}
\if\mydocumentclass2
The proof of the equivalence of \eqref{eq:phik2} and \eqref{eq:phik3} is provided in Appendix~D of the arXiv version~\cite{fullVer}.
\else
The proof of the equivalence of \eqref{eq:phik2} and \eqref{eq:phik3} is provided in Appendix~\ref{prove:phi_k_equivalence}.
\fi
From~\eqref{eq:phik3}, we can see that if $\tilde{w}_{\varphi_k}/\tilde{w}_{\psi_k} \ge 1$, $\phi_k^*(\bar{P}_k)$ is a continuous logarithmic function.
However, if $\tilde{w}_{\varphi_k}/\tilde{w}_{\psi_k} < 1$, it is a piecewise nonlinear function with a breakpoint at $\bar{P}_k=C_k^4$.
Nevertheless, we can show that it is a continuously differentiable concave function.

\begin{proposition}\label{thm:convexity_of_phi}
	The function $\phi_k^*$ in \eqref{eq:phik3} is a continuously differentiable concave function of $\bar{P}_k$ on $[0, \infty)$.
\end{proposition}

\begin{proof}
\if\mydocumentclass2
Due to the page limit, we leave the proof in Appendix~E of the arXiv version~\cite{fullVer}.
\else
See Appendix~\ref{prove:thm:convexity_of_phi}.
\fi
\end{proof}


By Proposition~\ref{thm:convexity_of_phi}, we can conclude that the objective function, $\phi^*$, of Problem~$\Master$ is a concave function of $\bar{\bP}$, and accordingly,  Problem~$\Master$ is a convex optimization  problem.
Hence, we can obtain its optimal solution using the Karush–Kuhn–Tucker (KKT) conditions~\cite{boyd2004convex}.

\if\mydocumentclass0
	\begin{theorem} \label{thm:Pkstar}
		The optimal solution, $\bar{\bP}^*=(\bar{P}_k^*)_{\forall k\in\mathcal{K}}$, to Problem~$\Master$ is provided as follows. For each Subchannel~$k\in\mathcal{K}$,
		\begin{equation} \label{Pkstar in thm}
			\bar{P}_k^* =
			\begin{dcases}
				\left[ \mu^* \tilde{w}_{\psi_k} B_k - {\eta_{k,\psi_k}} \right]_{0}^{\Pmaxk}, & \begin{multlined}[t]\textnormal{if } \tilde{w}_{\varphi_k}/\tilde{w}_{\psi_k} < 1 \\\textnormal{ and } \mu^* > C_k^5, \end{multlined}\\
				\left[ \mu^* \tilde{w}_{\varphi_k} B_k - {\eta_{k,\varphi_k}} \right]_{0}^{\Pmaxk}, & \textnormal{otherwise,}
			\end{dcases}
		\end{equation}
		where
		\begin{equation}
			C_k^5 = \frac{\eta_{k,\varphi_k}-\eta_{k,\psi_k}}{B_k(\tilde{w}_{\varphi_k}-\tilde{w}_{\psi_k})},
		\end{equation}
		and $\mu^*$ is chosen to satisfy $\sum_{k\in\mathcal{K}} \bar{P}_k^* = \Pmax$.
	\end{theorem}
\else
	\begin{theorem} \label{thm:Pkstar}
		The optimal solution, $\bar{\bP}^*=(\bar{P}_k^*)_{\forall k\in\mathcal{K}}$, to Problem~$\Master$ is provided as follows. For each Subchannel~$k\in\mathcal{K}$,
		\begin{equation} \label{Pkstar in thm}
			\bar{P}_k^* =
			\begin{dcases}
				\left[ \mu^* \tilde{w}_{\psi_k} B_k - {\eta_{k,\psi_k}} \right]_{0}^{\Pmaxk}, & \textnormal{if } \tilde{w}_{\varphi_k}/\tilde{w}_{\psi_k} < 1 \textnormal{ and } \mu^* > C_k^5,\\
				\left[ \mu^* \tilde{w}_{\varphi_k} B_k - {\eta_{k,\varphi_k}} \right]_{0}^{\Pmaxk}, & \textnormal{otherwise,}
			\end{dcases}
		\end{equation}
		where
		\begin{equation}
			C_k^5 = \frac{\eta_{k,\varphi_k}-\eta_{k,\psi_k}}{B_k(\tilde{w}_{\varphi_k}-\tilde{w}_{\psi_k})},
		\end{equation}
		and $\mu^*$ is chosen to satisfy $\sum_{k\in\mathcal{K}} \bar{P}_k^* = \Pmax$.
	\end{theorem}
\fi

\begin{proof}
\if\mydocumentclass2
Due to the page limit, we leave the proof in Appendix~F of the arXiv version~\cite{fullVer}.
\else
See Appendix~\ref{prove:thm:Pkstar}.
\fi
\end{proof}

Note that $\bar{P}_k^*$ is continuous, piecewise-linear, and increasing with respect to $\mu^*$.
Hence, there exists a unique solution, $\mu^*$, which can be easily found by any simple root-finding method such as a bisection method.
Once $\mu^*$ is determined, the optimal solution, $\bar{\bP}^*$, to Problem~$\Master$ can be obtained using~\eqref{Pkstar in thm}.
The pseudocode based on the bisection method is provided in Algorithm~\ref{Alg:Solve_Master}.
Finding the root of a single variable~$\mu^*$, as needed in our algorithm, in general, is much faster than standard convex programming tools required to solve Problem~$\Master$ consisting of $K$ variables and $2K+1$ constraints.

\begin{algorithm}[!t]
\DontPrintSemicolon
Let $f(\mu) = \sum_{k\in\mathcal{K}} \bar{P}_k^*(\mu) - \Pmax$, where $\bar{P}_k^*(\mu)$ is obtained by \eqref{Pkstar in thm} with $\mu^*=\mu$.\\
Initialize $\mu_L$ to zero, and $\mu_U$ to a sufficiently large value.\\
\Repeat{convergence}{
	$\mu_{\textnormal{new}} \gets (\mu_L+\mu_U)/2$.\\
	\lIf{$f(\mu_L) \cdot f(\mu_{\textnormal{new}}) < 0$}{
		$\mu_U \gets \mu_{\textnormal{new}}$.
	}
	\lIf{$f(\mu_U) \cdot f(\mu_{\textnormal{new}}) < 0$}{
		$\mu_L \gets \mu_{\textnormal{new}}$.
	}
}
Obtain $\bar{\bP}^*$ by plugging $\mu^* = (\mu_L+\mu_U)/2$ into \eqref{Pkstar in thm}.\\
\Return{$\bar{\bP}^*$.}
\caption{Solution to Problem~$\Master$}
\label{Alg:Solve_Master}
\end{algorithm}

We now analyze the computational complexity of Algorithm~\ref{Alg:Solve_Master}.
In each iteration, the computation of $f(\mu_{\textnormal{new}})$ dominates the others, and its computational complexity is $\mathcal{O}(K)$ due to the sum over $\mathcal{K}$.
Thus, letting $\kappa$ be the number of iteration until convergence, the overall computational complexity of Algorithm~\ref{Alg:Solve_Master} can be given by $\mathcal{O}(\kappa K)$.
Note that the exact value of $\kappa$ cannot be derived rigorously, but is usually considered to have an order of $\log K$~\cite{he2013water, khakurel2014generalized}.

Consequently, the total computational complexity of Algorithm~\ref{Alg:SG-Master}, running Algorithm~\ref{Alg:Solve_Slave} $K$ times and Algorithm~\ref{Alg:Solve_Master} once in each iteration, can be expressed as $\mathcal{O}(\xi K(N^2+\kappa))$, where $\xi$ is the number of iterations until convergence.

\begin{remark}
We can achieve the proportional fair scheduling by solving Problem~$\Pt$ in every time slot using the Joint-SAPA-LCC algorithm, where the effective weight of User~$i$ in time slot~$t$ is given by
\begin{equation}\label{eq:PFS_weight}
	\tilde{w}_i^t = \frac{1}{R_{\textnormal{EMA},i}^t},
\end{equation}
where $R_{\textnormal{EMA},i}^t$ is the exponential moving average data rate of User~$i$ in time slot $t$, and it can be recursively updated by
\begin{equation}\label{eq:PFS_avg_update}
	R_{\textnormal{EMA},i}^{t+1} = \left(1-\frac{1}{\tau}\right) R_{\textnormal{EMA},i}^t + \frac{1}{\tau} R_i^t,
\end{equation}
where $\tau$ is the time-averaging window coefficient.
\end{remark}

This remark shows that our Joint-SAPA-LCC algorithm can be easily extended to the proportional fair scheduling algorithm.
However, it is worth noting that the purpose of the proportional fair scheduling is to reduce the performance differences among all users, so it cannot guarantee the QoS constraints in~\eqref{Const:Ri}.
Hence, in the next section, we develop a scheduling algorithm that guarantees given QoS constraints by utilizing the Joint-SAPA-LCC algorithm.

\section{Opportunistic MC-NOMA Scheduling}
\label{sec:scheduling}
In this section, we finally develop an opportunistic MC-NOMA scheduling algorithm that works in an online manner by decomposing Problem~$\Pone$ into a series of Joint-SAPA problems over time slots, i.e., Problem~$\Pt$ for each time slot.
To this end, we first take advantage of the well-known property that if the fading process is stationary and ergodic, the long-term time average converges almost surely to the expectation for almost all realizations of the fading process~\cite{tse2005fundamentals, walters2000introduction}.
Thereby, by denoting a random vector representing the channel vector in a generic time slot by~$\bh$ and replacing the superscript~$t$ of the decision variables with $\bh$, we can reformulate Problem~$\Pone$ equivalently as
\begin{IEEEeqnarray}{c'c'l}
	\IEEEnonumber*
	\Ptwo & \maximize_{\bp^{\bh},\,\bq^{\bh},\,\forall \bh} & \E_{\bh} \left[ \sum_{i\in\calN} w_i R_i(\bp^\bh, \bq^\bh; \bh) \right] \label{Prob1:Obj}\\
	&\subjto
	&\E_{\bh} \left[ R_i(\bp^{\bh}, \bq^{\bh} ; \bh) \right] \ge \bar{R}_{\textnormal{min},i},~\forall i\in\calN, \label{Prob1:minR}\\
	&&\bp^\bh \in \mathcal{P}, ~ \bq^\bh \in \mathcal{Q}, ~ \forall \bh. \label{Prob1:constraint for P and Q}
\end{IEEEeqnarray}
At each time slot~$t$ where the channel vector is realized as $\bh^t$, subchannel assignment and power allocation can be done according to the solution for $\bp^\bh$ and $\bq^\bh$ obtained by solving Problem~$\Ptwo$ with $\bh=\bh^t$.

There is still a big challenge in solving Problem~$\Ptwo$.
Since no information on the underlying distributions of the fading process is provided, we have to solve the stochastic optimization problem without such information.
To resolve it, we leverage the Lagrangian duality and the stochastic optimization theory to develop the opportunistic MC-NOMA scheduling algorithm.
Its core mechanism is to take advantage of the time-varying channel conditions opportunistically to maximize the weighted average sum rate.
Also, the effective weights are systemically adjusted so that the QoS requirements (i.e., the individual minimum average data rate requirements) are fulfilled.
To develop the algorithm, by introducing a Lagrange multiplier, $\lambda_i$, for the minimum average data rate constraint of User~$i$, we first define a Lagrangian function, $L$, associated with Problem~$\Ptwo$ as
\if\mydocumentclass0
	\begin{align}\label{eq:lagrangian}
		L(\barbp, \barbq, \blambda)
		&=\begin{multlined}[t]
			\E_{\bh} \left[ \sum_{i\in\calN} w_i R_i(\bp^{\bh}, \bq^{\bh} ; \bh) \right] \\
			+ \sum_{i\in\calN} \lambda_i \left( \E_{\bh} \left[ R_i(\bp^{\bh}, \bq^{\bh} ; \bh) \right] - \bar{R}_{\textnormal{min},i} \right)
		\end{multlined} \nonumber\\
		&= \E_{\bh} \left[ \sum_{i\in\calN} (w_i+\lambda_i) R_i(\bp^{\bh}, \bq^{\bh} ; \bh) \right] - \sum_{i\in\calN} \lambda_i \bar{R}_{\textnormal{min},i},
	\end{align}
\else
	\begin{align}\label{eq:lagrangian}
		L(\barbp, \barbq, \blambda)
		&= \E_{\bh} \left[ \sum_{i\in\calN} w_i R_i(\bp^{\bh}, \bq^{\bh} ; \bh) \right] + \sum_{i\in\calN} \lambda_i \left(\E_{\bh}\left[ R_i(\bp^{\bh}, \bq^{\bh} ; \bh) \right] - \bar{R}_{\textnormal{min},i}\right) \nonumber\\
		&= \E_{\bh} \left[ \sum_{i\in\calN} (w_i+\lambda_i) R_i(\bp^{\bh}, \bq^{\bh} ; \bh) \right] - \sum_{i\in\calN} \lambda_i \bar{R}_{\textnormal{min},i},
	\end{align}
\fi
where $\barbp = (\bp^{\bh})_{\forall\bh}$, $\barbq = (\bq^{\bh})_{\forall\bh}$, and $\blambda=(\lambda_i)_{\forall i\in\mathcal{N}}$.
Then, the dual problem associated with Problem~$\Ptwo$ can be defined by
\begin{IEEEeqnarray}{c'c'l}\label{Prob:Dual}
	\IEEEnonumber*
	\Dual & \minimize_{\blambda} & F(\blambda) \\
	&\subjto & \lambda_i \ge 0, ~ \forall i\in\mathcal{N},
\end{IEEEeqnarray}
where
\begin{IEEEeqnarray}{rCl'l}
	\IEEEyesnumber\label{Prob:F}
	F(\blambda) &=& \maximize_{\bp^\bh\in\mathcal{P},\,\bq^\bh\in\mathcal{Q},\,\forall \bh} & L(\barbp, \barbq, \blambda).
\end{IEEEeqnarray}
Since Problem~$\Ptwo$ is nonconvex, even though its dual problem, Problem~$\Dual$, is optimally solved, there may be a duality gap.
However, the duality gap vanishes in our problem, resulting in no loss of optimality.



\begin{theorem}\label{thm:zero-duality-gap}
The strong duality (i.e., zero duality gap) holds between Problem~$\Ptwo$ and its dual problem, Problem~$\Dual$.
\end{theorem}

\begin{proof}
\if\mydocumentclass2
Due to the page limit, we leave the proof in Appendix~G of the arXiv version~\cite{fullVer}.
\else
See Appendix~\ref{prove:thm:zero-duality-gap}.
\fi
\end{proof}

We thus develop an algorithm that solves Problem~$\Dual$.
To this end, we first focus on obtaining its objective function, $F(\blambda)$.
The first (expectation) term in~\eqref{eq:lagrangian} is separable for each channel realization, and the second term is independent of the decision variables, $\barbp$ and $\barbq$.
Hence, for a given Lagrange multiplier vector, $\blambda$, the maximization in~\eqref{Prob:F} can be solved by separately solving the subproblem for each channel realization, defined by
\begin{IEEEeqnarray}{c'c'l}\label{Prob:Dual_s}
	\IEEEnonumber*
	\Dualh & \maximize_{\bp^\bh\in\mathcal{P},\,\bq^\bh\in\mathcal{Q}} & \sum_{i\in\calN} (w_i+\lambda_i) R_i(\bp^\bh, \bq^\bh ; \bh).
\end{IEEEeqnarray}
Since the expectation has disappeared in Problem~$\Dualh$, it can be solved without knowledge of the underlying distributions of the fading process once the channel realization is provided.
Thus, for any given $\blambda$ and $\bh$, Problem~$\Dualh$ turns into a deterministic optimization problem for Joint-SAPA that aims to maximize the instantaneous weighted sum rate with weight $w_i+\lambda_i$ for User~$i$.
This problem can be solved using the Joint-SAPA-LCC algorithm developed in the previous section with letting $\tilde{w}_i = w_i + \lambda_i$.

We now focus back on solving Problem~$\Dual$.
Even though Problem~$\Dualh$ can be solved for given $\bh$ and $\blambda$, the underlying distributions of the fading process are still fundamentally needed to solve Problem~$\Dual$.
Nevertheless, thanks to the fact that Problem~$\Dual$ is a form of convex stochastic optimization problems, we can solve it without resorting on the distributions using the stochastic subgradient method~\cite{shapiro2014lectures}, where the Lagrange multiplier vector, $\blambda$, is iteratively updated by
\begin{equation}\label{eq:update}
	\blambda^{t+1} = \left[\blambda^{t} - \zeta^{t} \bv^{t}\right]^+,
\end{equation}
where $\blambda^{t}$ and $\zeta^{t}$ are the Lagrange multiplier vector and the positive step size in time slot~$t$, respectively, and $\bv^t=(v_i^t)_{\forall i\in\mathcal{N}}$ is the stochastic subgradient of $F(\blambda)$ with respect to $\blambda$ at $\blambda=\blambda^t$.
By Danskin's min-max theorem~\cite{bertsekas1999nonlinear}, the stochastic subgradient, $\bv^t$, can be obtained by
\begin{equation}\label{eq:v}
	v_i^t = R_i^t - \bar{R}_{\textnormal{min},i}, ~ \forall i \in \calN,
\end{equation}
where $R_i^t$ is the instantaneous data rate of User~$i$ in time slot~$t$ defined as in~\eqref{eq:Rate}, which is achieved when the subchannel assignment and power allocation are performed according to the solution to Problem~$\Dualh$ with $\bh=\bh^t$ and $\blambda=\blambda^t$.
With the update process of \eqref{eq:update}, the Lagrange multiplier vector converges almost surely to the optimal solution, $\blambda^*$, of Problem~$\Dual$ if the step size~$\zeta^t$ is square-summable, but not summable~\cite{boyd2008stochastic}, i.e.,
\begin{equation}\label{eq:condition}
\zeta^{t} \ge 0, ~ \sum_{t=1}^{\infty} \zeta^{t} = \infty, ~ \textnormal{and} ~ \sum_{t=1}^{\infty} (\zeta^{t})^2 < \infty.
\end{equation}
The proposed algorithm for the opportunistic MC-NOMA scheduling is outlined in Algorithm~\ref{Alg:Scheduling}.

\begin{algorithm}[!t]
\DontPrintSemicolon
Initialize: $t=1$ and $\blambda^{t}=\mathbf{0}$.\\
\For{each time slot~$t$}{
	Solve Problem~$\Dualh$ with $\bh=\bh^t$ and $\blambda=\blambda^t$ using the Joint-SAPA-LCC algorithm (Algorithm~\ref{Alg:SG-Master}).\\
	Transmit the signal generated by the obtained solution.\\
	Calculate $\blambda^{t+1}$ according to~\eqref{eq:update} and~\eqref{eq:v}.\\
	$t \gets t+1$.\\
}
\caption{Opportunistic MC-NOMA Scheduling}
\label{Alg:Scheduling}
\end{algorithm}

It is worth noting that, due to the stationary ergodic fading process, Problem~$\Pone$ and Problem~$\Ptwo$ are equivalent with probability one, and according to Theorem~\ref{thm:zero-duality-gap}, there is no duality gap between Problem~$\Ptwo$ and Problem~$\Dual$.
Therefore, there is no loss of optimality in solving Problem~$\Pone$ by our algorithms as long as Problem~$\Dualh$, i.e., Problem~$\Pt$, is optimally solved.
Accordingly, we can expect that our opportunistic MC-NOMA scheduling provides near-optimal performance by showing via simulation that the Joint-SAPA-LCC algorithm provides near-optimal performance.

%


\section{Simulation Results}
\label{sec:sim}
Now, we present simulation results to evaluate the performance of our proposed algorithms.
We first investigate the Joint-SAPA-LCC algorithm in Section~\ref{sec:sim1} and then the opportunistic MC-NOMA scheduling algorithm in Section~\ref{sec:sim2}.
Throughout the simulations, we consider a circular cell with a radius of \SI{300}{\meter}, in which one BS located at the center of the cell serves $N$ users over $K$ subchannels.
The system bandwidth, $B_{\textnormal{tot}}$, is set to \SI{5}{\MHz}, and the bandwidth of each subchannel is set equally to $B_{\textnormal{tot}}/K$.
Unless otherwise specified, the numbers of users and subchannels are set to $10$, i.e., $N=10$ and $K=10$.
The total transmission power budget, $\Pmax$, of the BS is set to \SI{43}{\dBm}, and the transmission power budget for each subchannel is set to $\gamma\Pmax/K$ with $\gamma=1.15$.
The large-scale path loss is modeled by the HATA model for urban environments~\cite{hata1980empirical, access13radio}.
Specifically, the path loss in dB over distance $d_{km}$ in kilometers is set to $69.55+26.16\log_{10}(f_c) - a(h_m, h_b) + b(h_b) \log_{10}(d_{km})$, where $a(h_m, h_b) = 13.82\log_{10}(h_b) + 3.2 [ \log_{10}(11.75 h_m)]^2 - 4.97$, $b(h_b) = 44.9-6.55\log_{10}(h_b)$, $f_c$ is the carrier frequency, and $h_m$ and $h_b$ are effective antenna heights of the BS and users, respectively.
The parameters are set as follows: $f_c=900$\,\si{\MHz}, $h_a=30$\,\si{\meter}, and $h_b=2$\,\si{\meter}.
Also, we set the antenna gains of the BS and the users to \SI{15}{\dBi} and \SI{0}{\dBi}, respectively.
Then, we consider the shadow fading with a standard deviation of \SI{8}{\dB} and the Rayleigh small-scale fading with unit variance.
The noise power spectral density, $N_0$, is set to \SI[per-mode=symbol]{-174}{\dBm\per\Hz}.
Thus, the noise power of User~$i$ on Subchannel~$k$ is given by $\sigma_{k,i}^2=B_{\textnormal{tot}}N_0/K$.
In the following simulation results, the data rates are considered in units of \si[per-mode=symbol]{\bps\per\Hz}.
In Algorithm~\ref{Alg:Scheduling}, the step size in~\eqref{eq:update} is set to $\zeta^t = 1/t$, which satisfies the conditions in~\eqref{eq:condition} so that the convergence of the algorithm is guaranteed.

\subsection{Joint Subchannel Assignment and Power Allocation}
\label{sec:sim1}
In this subsection, we provide the performance of our Joint-SAPA-LCC algorithm that aims to maximize the weighted sum rate by solving Problem~$\Pt$.
For comparison, we additionally provide the performance of three other algorithms.
The first one is the Joint-SAPA-FTPC algorithm~\cite{saito2013system}, where Joint-SAPA is performed based on the FTPC and ES algorithms.
The second one is the Joint-SAPA-DCP algorithm~\cite{parida2014power}, where Joint-SAPA is performed based on the DCP approach.
The last one is the Joint-SAPA-DP algorithm~\cite{salaun2020joint}, where Joint-SAPA is performed based on the DP approach and the PGD method.
In the following simulations, we assume that $N$ users are uniformly distributed within the circular cell with at least \SI{30}{\meter} away from the BS, and their weights are randomly set between $0$ and $1$.
Also, taking into account the high computational complexity of the above baseline algorithms, we assume that each subchannel can be assigned to up to $5$ users, i.e., $M=5$.
All the simulation results are averaged over $3000$ independent trials.
For each trial, locations, weights, and channel gains of all users are independently generated.
Note that the Joint-SAPA algorithms (including ours) do not deal with any QoS constraints.
The purpose of this subsection is to verify that our Joint-SAPA-LCC algorithm provides good enough performance despite requiring very low computational complexity compared to the baseline Joint-SAPA algorithms.

We first compare the computational complexity of the Joint-SAPA algorithms.
To this end, we define the relative computational cost of an algorithm as its execution time normalized to that of our Joint-SAPA-LCC algorithm, and then show the corresponding results for different numbers of users and subchannels in Figs.~\ref{fig:compare:complex_N} and~\ref{fig:compare:complex_K}, respectively.
The execution times were measured by MATLAB R2020a software on a computer with Intel Core i7-9700K CPU (3.60 GHz) and 32.0 GB RAM.
From the figures, we can see that our Joint-SAPA-LCC algorithm is much faster than the other algorithms.
For example, when $N=10$ and $K=10$, the computational cost of Joint-SAPA-LCC is about $100$, $250$, and $550$ times lower than those of Joint-SAPA-FTPC, Joint-SAPA-DCP, and Joint-SAPA-DP, respectively.
The main reason why our Joint-SAPA-LCC algorithm is fast is that the \textit{candidate} users who might be allocated positive power are determined simply based on the closed-form power allocation formulas.
Furthermore, although not proven theoretically, we were able to observe experimentally that Algorithm~\ref{Alg:SG-Master} converges in only a single iteration in most cases.
That is, in most cases, the Joint-SAPA-LCC algorithm is performed in a 3-step procedure: i) to obtain $\phi^*$ by selecting \textit{candidate} users based on equal subchannel power allocation, ii) to refine the subchannel power allocation~$\bar{\bP}$ based on $\phi^*$ obtained in the first step, and iii) to obtain the final Joint-SAPA solution based on $\bar{\bP}$ obtained in the second step.
This computational complexity comparison confirms that our Joint-SAPA-LCC algorithm is very effective and well suited to be implemented in practical systems where Joint-SAPA should be performed in every very short time slot.


\begin{figure}[!t]
\centering
\includegraphics[width=8cm]{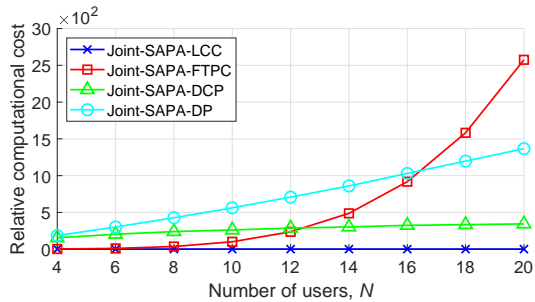}
\caption{The relative computational cost versus the number of users.}\label{fig:compare:complex_N}
\end{figure}

\begin{figure}[!t]
\centering
\includegraphics[width=8cm]{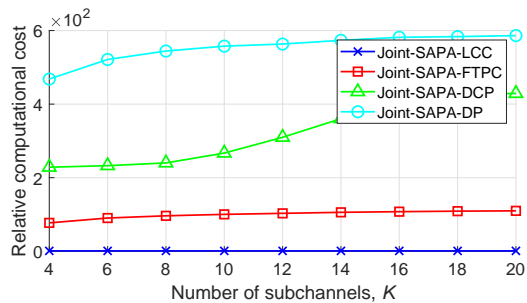}
\caption{The relative computational cost versus the number of subchannels.}\label{fig:compare:complex_K}
\end{figure}

In Figs.~\ref{fig:compare:wsr_N} and \ref{fig:compare:wsr_K}, we compare the weighted sum rate performance with varying the numbers of users and subchannels, respectively.
First, as shown in Fig.~\ref{fig:compare:wsr_N}, as the number of users increases, the weighted sum rates increase in all the Joint-SAPA algorithms thanks to the increase in the multi-user diversity gain.
Also, we can see that our Joint-SAPA-LCC algorithm, despite its very low computational complexity, has only a little performance drop compared to the Joint-SAPA-DP and Joint-SAPA-DCP algorithms, and provides higher performance compared to the Joint-SAPA-FTPC algorithm.
On the other hand, Fig.~\ref{fig:compare:wsr_K} shows that the weighted sum rates for all the Joint-SAPA algorithms tend to remain constant regardless of the number of subchannels.
These results indicate that, as also seen in~\cite{jang2003transmit, kim2010joint}, the effects of the number of subchannels on the weighted sum rate performance are negligible in the system where the transcievers for different subchannels operate independently.
Meanwhile, the order between the Joint-SAPA algorithms in terms of the weighted sum rate performance remains the same as in Fig.~\ref{fig:compare:wsr_N}.
For example, when $N=10$ and $K=10$, the weighted sum rate of our Joint-SAPA-LCC algorithm is only \SI{0.86}{\percent} and \SI{0.12}{\percent} lower than those of the Joint-SAPA-DP and Joint-SAPA-DCP algorithms, respectively, but \SI{4}{\percent} higher than that of the Joint-SAPA-FTPC algorithm.


\begin{figure}[!t]
\centering
\includegraphics[width=8cm]{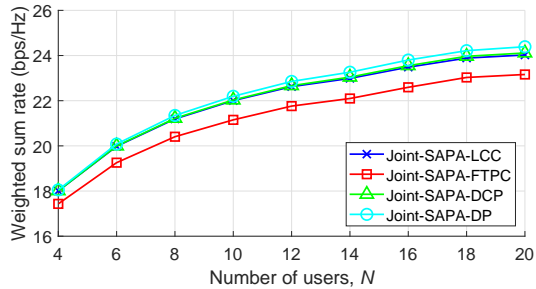}
\caption{The weighted sum rate versus the number of users.}
\label{fig:compare:wsr_N}
\end{figure}

\begin{figure}[!t]
\centering
\includegraphics[width=8cm]{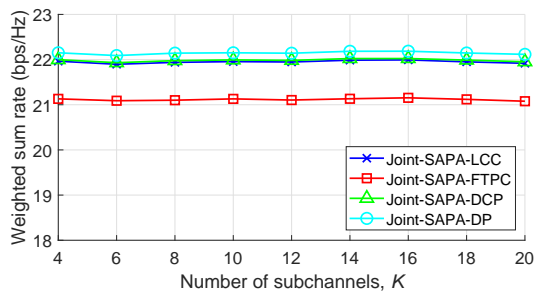}
\caption{The weighted sum rate versus the number of subchannels.}
\label{fig:compare:wsr_K}
\end{figure}

Next, since the Joint-SAPA-DP algorithm provides the highest weighted sum rate performance, we compare our Joint-SAPA-LCC algorithm with it in more depth for the case where $N=10$ and $K=10$ in Fig.~\ref{fig:compare:Salaun}.
Fig.~\ref{fig:compare:Hist_number_of_selected_users} shows the frequency histogram of the number of assigned users per subchannel.
As can be seen in the figure, at most five users are assigned per subchannel in the Joint-SAPA-DP algorithm, unlike our algorithm with at most two users per subchannel.
As a result, the user assignment patterns of the two algorithms are different, and the Joint-SAPA-DP algorithm has a larger solution space for user assignment patterns compared to our algorithm.
Accordingly, the Joint-SAPA-DP algorithm has the possibility to provide higher performance than ours.
On the other hand, Fig.~\ref{fig:compare:wdr} shows the achieved weighted data rates of individual assigned users in a subchannel.
In this figure, the value of the bar for Index~$i$ represents the average weighted data rate of the user that achieves the $i$th highest data rate among the assigned users in a subchannel.
We can see that more than \SI{95}{\percent} of the weighted sum rate performance is assigned to the top two indices in the Joint-SAPA-DP algorithm, which implies that the remaining bottom three users have a very little impact on the weighted sum rate performance.
Since most of the weighted sum rates correspond to the first two indices, whether in our Joint-SAPA-LCC algorithm or in the Joint-SAPA-DP algorithm, the weighted sum rate performance of our Joint-SAPA-LCC algorithm is very close to that of the Joint-SAPA-DP algorithm despite the different user assignment patterns.
In summary, the simulation results thus far confirm that not only does our Joint-SAPA-LCC algorithm provide good performance close to that of the Joint-SAPA-DP algorithm, but it also has much lower computational complexity compared to the other Joint-SAPA algorithms, which is critical for implementation in practical systems.

\begin{figure}[!t]
\centering
\subfloat[The frequency histogram of the number of assigned users per subchannel.]{\includegraphics[width=7cm]{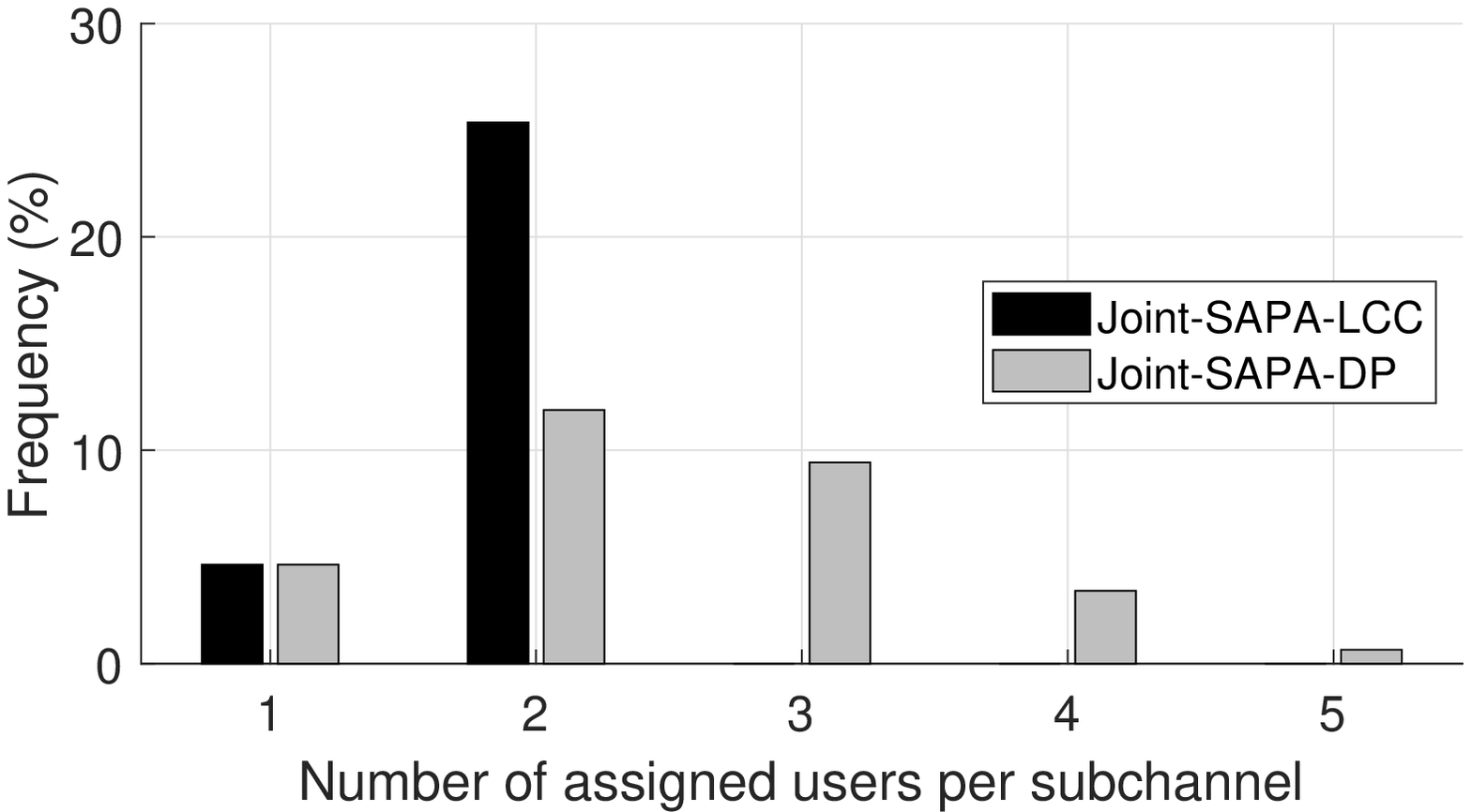}\label{fig:compare:Hist_number_of_selected_users}}
\hfil
\subfloat[The weighted data rates in descending order.]{\includegraphics[width=7cm]{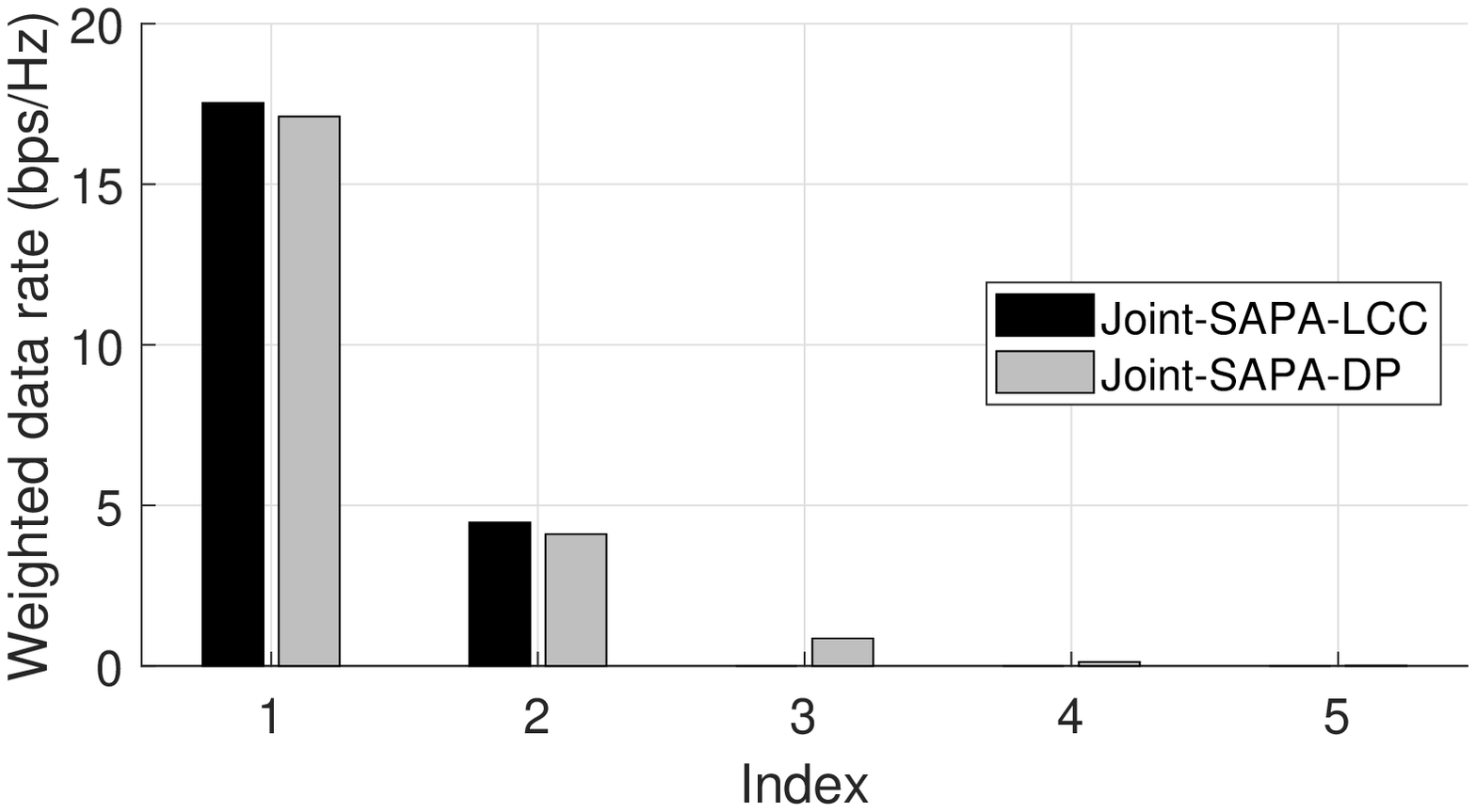}\label{fig:compare:wdr}}
\caption{The comparison between the Joint-SAPA-LCC and Joint-SAPA-DP algorithms for the case where $N=10$ and $K=10$.\label{fig:compare:Salaun}}
\end{figure}

Lastly, we further investigate the effects of the imperfection of instantaneous channel state information (CSI) on the weighted sum rate performance since perfect knowledge of CSI at the BS is practically impossible.
To this end, as in~\cite{wei2017optimal, zhu2021joint}, we first model an estimated channel gain corresponding to $h_{k,i}$ as $\hat{h}_{k,i}=h_{k,i}+e_{k,i}$, where $e_{k,i}$ is the channel estimation error generated by the complex Gaussian distribution with zero mean and variance of $\sigma_e^2/\textnormal{PL}_i$, and $\textnormal{PL}_i$ is the path loss of User~$i$.
Fig.~\ref{fig:compare:err} shows the weighted sum rate performance of our Joint-SAPA-LCC algorithm in the imperfect CSI environments.
The parameter settings for Figs.~\ref{fig:compare:wsr_N_err} and \ref{fig:compare:wsr_K_err} are the same as for Figs.~\ref{fig:compare:wsr_N} and \ref{fig:compare:wsr_K}, respectively, only except that the imperfect CSI environments are applied.
As expected, the weighted sum rates are slightly degraded as the channel estimation error variance increases because the imperfect CSI distorts subchannel assignment (i.e., user pairing per subchannel) as well as power allocation.
Nevertheless, the performance degradation is not much, which means that the Joint-SAPA-LCC algorithm can tolerate some degree of channel estimation errors.

\begin{figure}[!t]
\centering
\subfloat[The weighted data rates with varying $N$ and $\sigma_e^2$.]{\includegraphics[width=7cm]{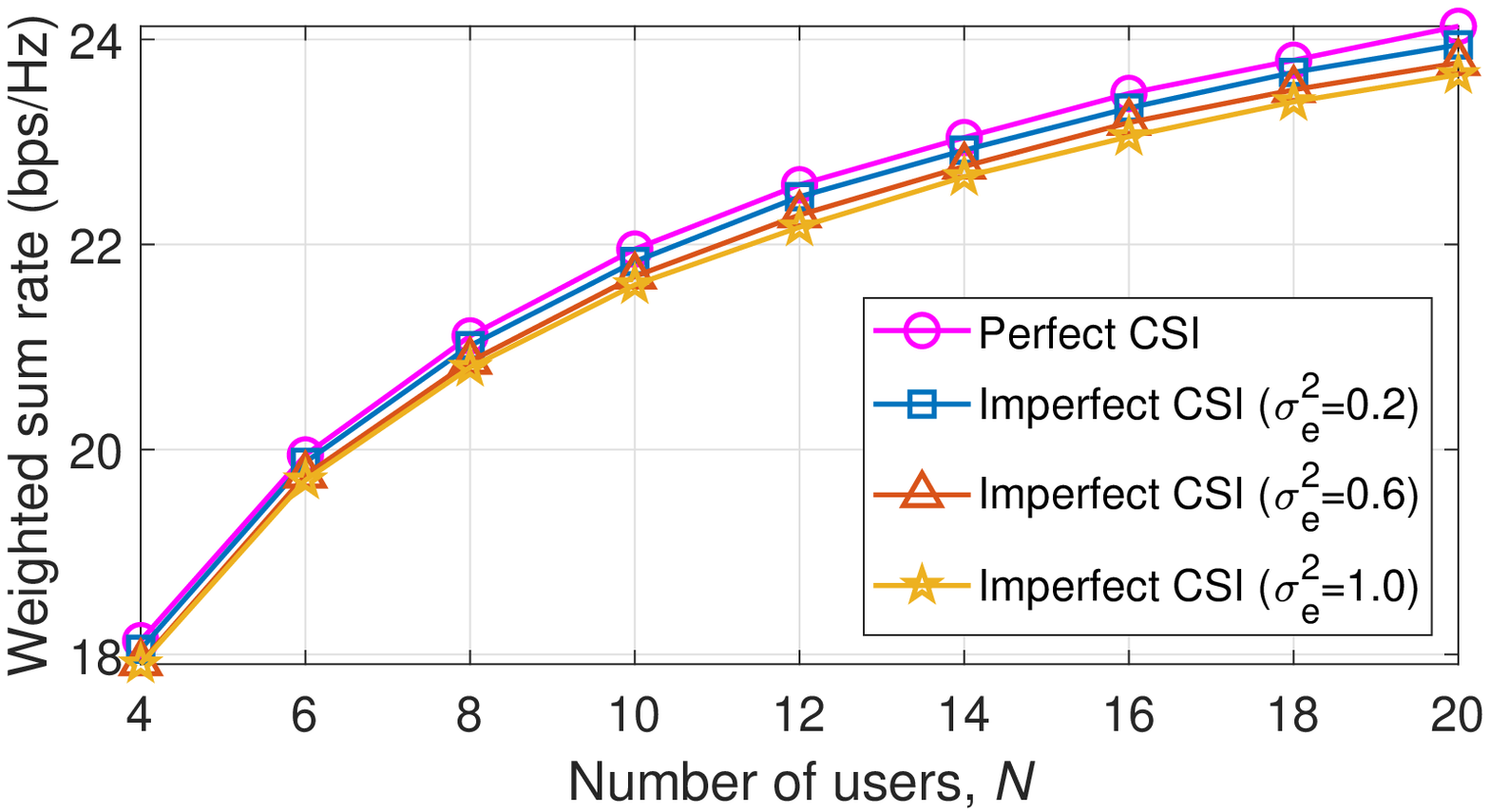}\label{fig:compare:wsr_N_err}}
\hfil
\subfloat[The weighted data rates with varying $K$ and $\sigma_e^2$.]{\includegraphics[width=7cm]{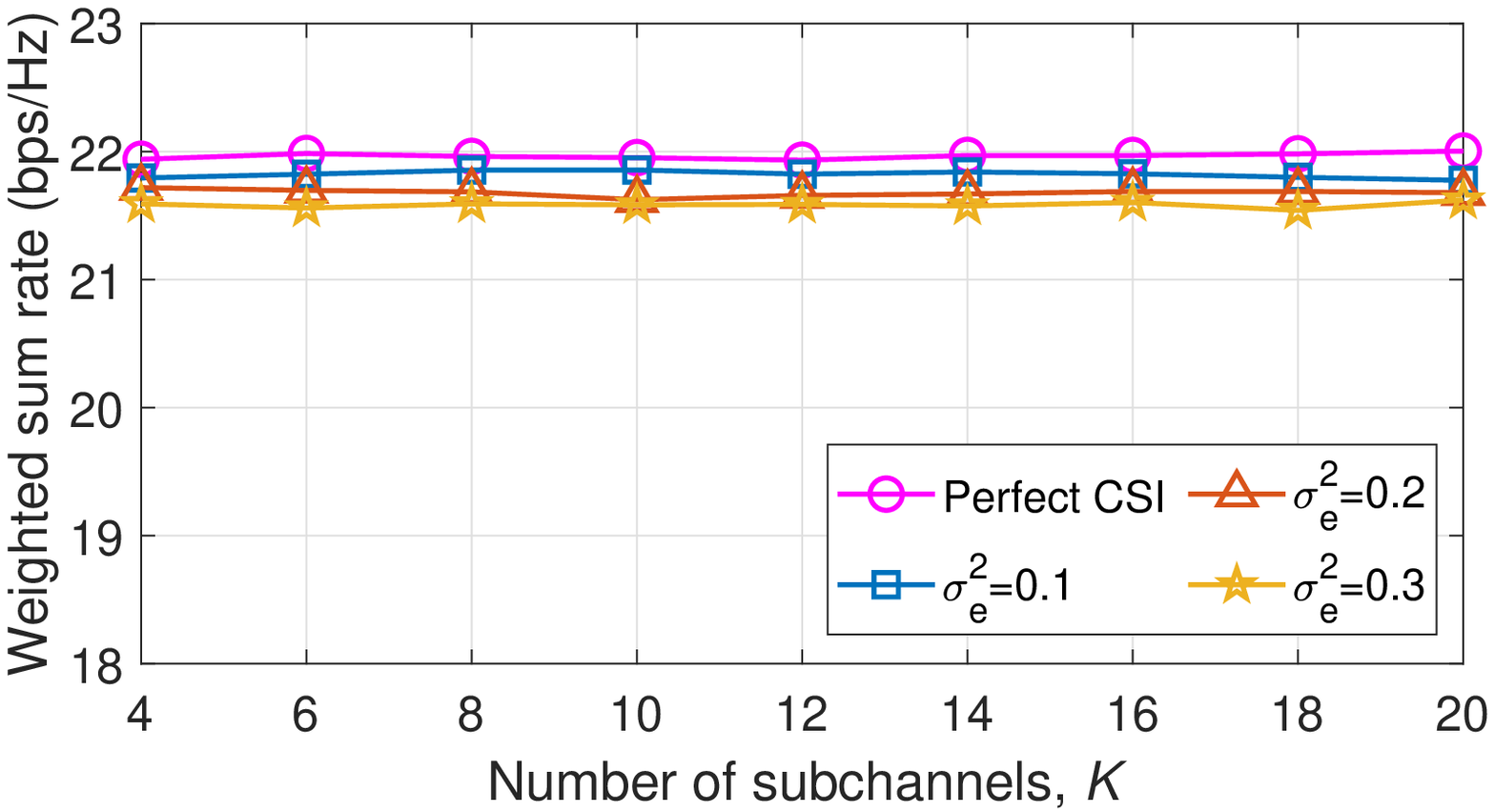}\label{fig:compare:wsr_K_err}}
\caption{The weighted sum rates of the Joint-SAPA-LCC algorithm under the imperfect CSI environments.\label{fig:compare:err}}
\end{figure}

\subsection{Opportunistic MC-NOMA scheduling}
\label{sec:sim2}
In this subsection, we provide the performance of our opportunistic MC-NOMA scheduling algorithm, taking into account the time-varying and frequency-selective channel conditions and various QoS requirements, i.e., the required minimum average data rates of users.
To show the effectiveness of our MC-NOMA scheduling algorithm, we compare its simulation results with those of two other scheduling algorithms: MC-NOMA scheduling without QoS requirements and proportional fair scheduling.
To be specific, the MC-NOMA scheduling without QoS requirements is achieved by solving Problem~$\Pone$ using our opportunistic MC-NOMA scheduling algorithm, where $\bar{R}_{\textnormal{min},i}$ is set to $0$ for all $i\in\mathcal{N}$.
Meanwhile, the proportional fair scheduling is achieved by solving Problem~$\Pt$ using our Joint-SAPA-LCC algorithm at each time slot, where the weight of each user is given as the reciprocal of its time-averaged data rate up until to that time slot as in~\eqref{eq:PFS_weight}.
The time-averaging window coefficient, $\tau$, in~\eqref{eq:PFS_avg_update} is set to $1000$.
Throughout the following simulations, we consider a system where there are $10$ users with equal weights in the cell, and the $i$th user is $30\times i$\,\si{\meter} away from the BS.
Thus, the lower the user index, the closer it is to the BS, resulting in a higher channel gain on average.
The performance results of the scheduling algorithms are investigated in the following two scenarios.
In one scenario, we assume that all users have the same QoS requirements.
Specifically, the minimum average data rates of all the users are set to \SI[per-mode=symbol]{2}{\bps\per\Hz}.
In the other scenario, we assume that the users have individually different QoS requirements.
Specifically, the minimum average data rates of Users~$1$, $2$, $5$, $6$, $9$, and $10$ are set to \SI[per-mode=symbol]{3.5}{\bps\per\Hz}, while those of Users~$3$, $4$, $7$, and $8$ are set to \SI[per-mode=symbol]{1}{\bps\per\Hz}.

\begin{figure}[!t]
\centering
\subfloat[The average sum rate.]{\includegraphics[width=7cm]{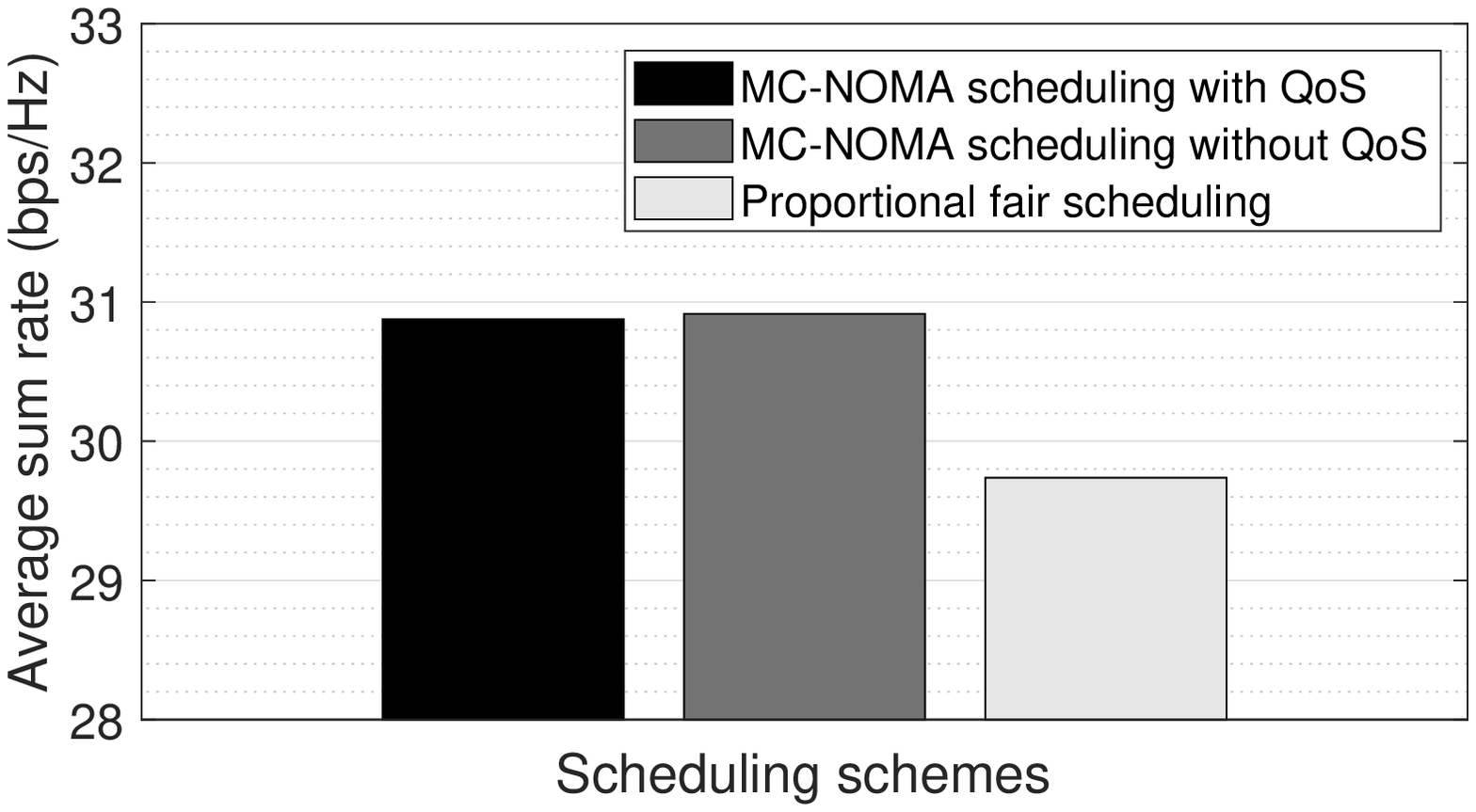}\label{fig:scheduling:equal_QoS:WSR}}
\hfil
\subfloat[The average data rate of each users.]{\includegraphics[width=7cm]{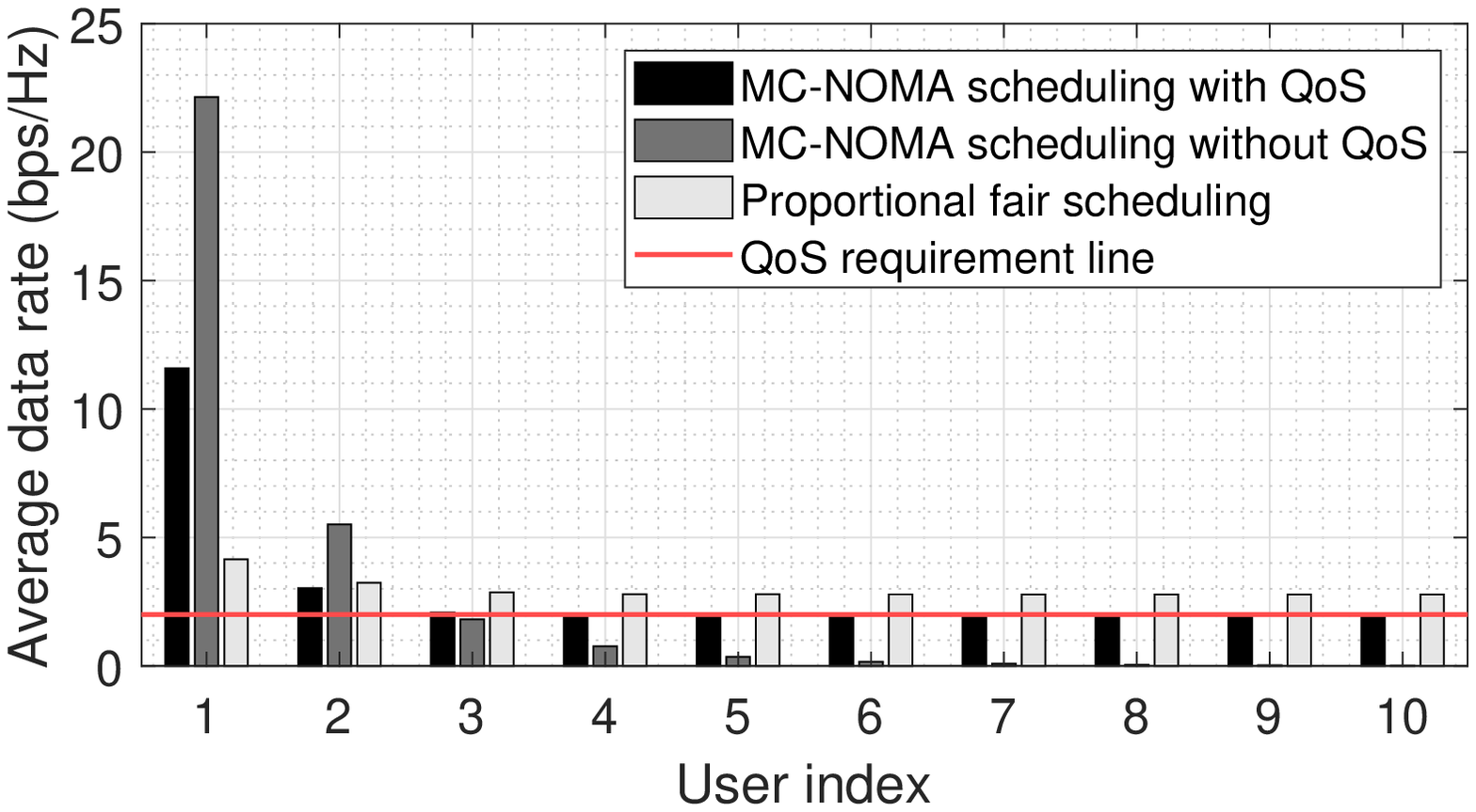}\label{fig:scheduling:equal_QoS:WDR}}
\caption{Performance comparison results between scheduling algorithms when all users have the same QoS requirements.\label{fig:scheduling:equal_QoS}}
\end{figure}

\begin{figure}[!t]
\centering
\subfloat[The average sum rate.]{\includegraphics[width=7cm]{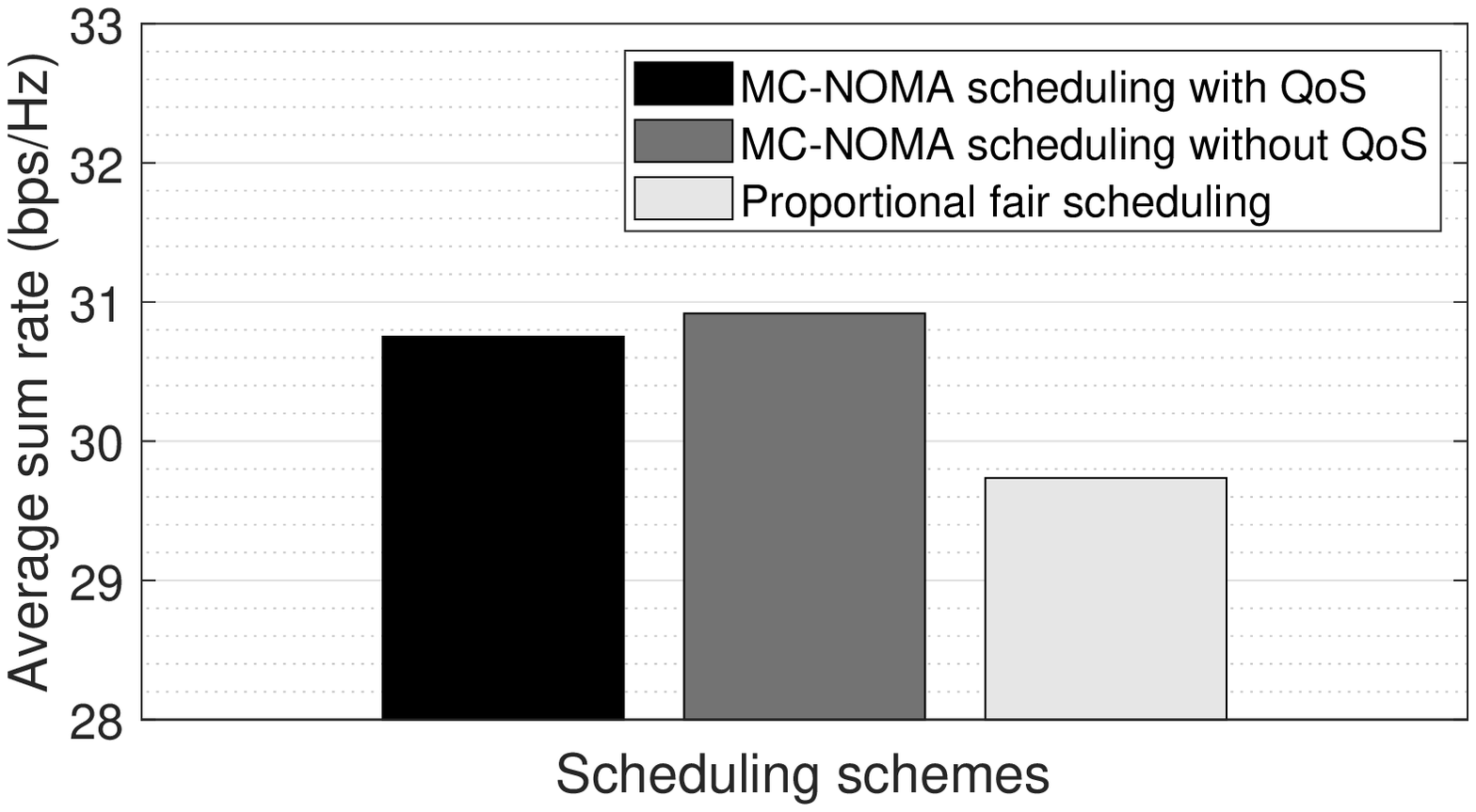}\label{fig:scheduling:diff_QoS:WSR}}
\hfil
\subfloat[The average data rate of each users.]{\includegraphics[width=7cm]{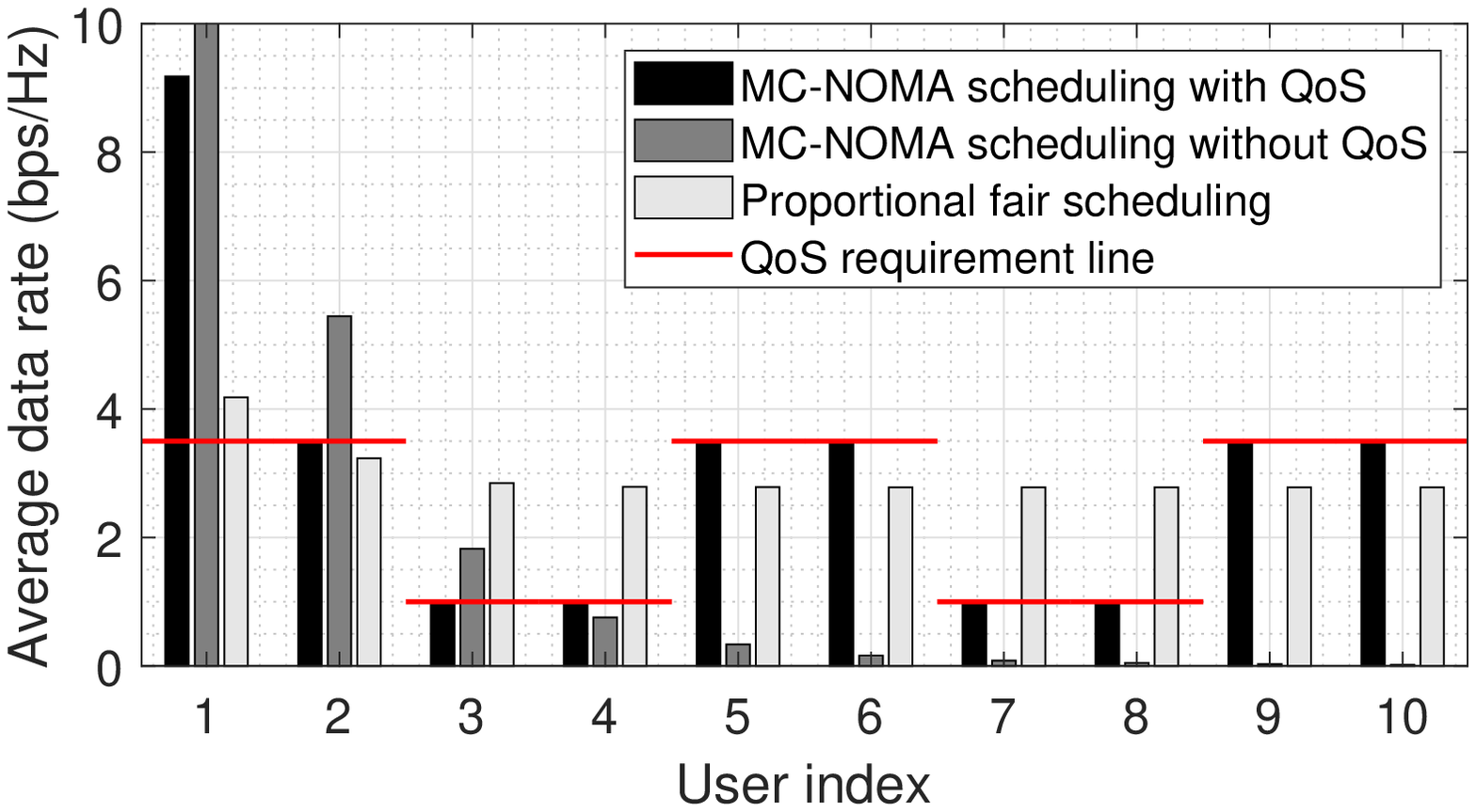}\label{fig:scheduling:diff_QoS:WDR}}
\caption{Performance comparison results between scheduling algorithms when users have individually different QoS requirements.\label{fig:scheduling:diff_QoS}}
\end{figure}

Figs.~\ref{fig:scheduling:equal_QoS} and \ref{fig:scheduling:diff_QoS} show the performance results for the first and second scenarios, respectively.
In Figs.~\ref{fig:scheduling:equal_QoS:WSR} and~\ref{fig:scheduling:diff_QoS:WSR}, we can see that the average sum rates of the MC-NOMA scheduling without QoS requirements are slightly higher than those of the MC-NOMA scheduling with QoS requirements in both cases.
As can be expected, this result is obvious because the feasible space of Problem~$\Pone$ with positive $\bar{R}_{\textnormal{min},i}$'s is a subspace of that of Problem~$\Pone$ with zero $\bar{R}_{\textnormal{min},i}$'s.
However, the lack of the QoS requirements makes the performance of users with poor channel conditions compromised to maximize the sum rate.
Consequently, as can be seen in Figs.~\ref{fig:scheduling:equal_QoS:WDR} and~\ref{fig:scheduling:diff_QoS:WDR}, only a few users close to the BS exploit the resources exclusively, and thereby users far from the BS do not meet their QoS requirements.
On the contrary, the QoS requirements of all users are satisfied well in the MC-NOMA scheduling with QoS requirements.
Meanwhile, the proportional fair scheduling follows the principle of giving high effective weights to users with low time-averaged data rates.
Accordingly, not only does it prevent users from starvation, but it also provides similar average data rate performance among users to gratify its purpose of maximizing the fairness utility function.
However, since it concentrates on the fairness between users and does not take into account the QoS requirements explicitly, a situation where the given QoS requirements are not satisfied may occur.
Fig.~\ref{fig:scheduling:diff_QoS:WDR} shows the case where the proportional fair scheduling cannot meet the QoS requirements, whereas our MC-NOMA scheduling with QoS requirements well satisfies them.
%
In summary, the simulation results demonstrate that our proposed scheduling algorithm not only provides good performance, but also guarantees given QoS requirements.



\section{Conclusion}
\label{sec:conc}
In this paper, we have studied the dynamic resource scheduling problem for joint user, subchannel, and power scheduling in the downlink MC-NOMA system over time-varying fading channels, which has the goal of maximizing the weighted average sum rate while ensuring given QoS requirements.
To this end, we have first developed the Joint-SAPA-LCC algorithm to maximize the instantaneous weighted sum rate.
By its characteristic that it leverages very simple user selection and power allocation based on closed-form equations, we could achieve much lower computational complexity compared to the existing Joint-SAPA algorithms.
In succession, along with the proposal of the proportional fair scheduling based on our Joint-SAPA-LCC algorithm, we have developed the opportunistic MC-NOMA scheduling algorithm that systematically adjusts the effective weights so that the weighted average sum rate is maximized while the QoS requirements are met.
Then, through the extensive simulation results, we have demonstrated that our Joint-SAPA-LCC algorithm provides good performance comparable to the Joint-SAPA-DP algorithm despite its much lower computational complexity, and that our opportunistic MC-NOMA scheduling algorithm satisfies given QoS requirements.
As a final remark, the issue of SIC error propagation has not been discussed in this paper.
We leave this issue for our future work.
This study will be the cornerstone for our future work on the development of scheduling for more complex systems, such as multi-cell MC-NOMA systems and massive MIMO MC-NOMA systems.

\if\mydocumentclass2

\else

\appendices
\numberwithin{equation}{section}
\numberwithin{definition}{section}
\numberwithin{lemma}{section}
\section{Proof of Theorem~\ref{thm:convergence}}
\label{prove:thm:convergence}

First, the objective value of Problem~$\Pteq$, i.e., the weighted sum rate over all subchannels, is bounded above since the total transmission power budget of the BS is limited.
Hence, by the monotone convergence theorem, we can guarantee that Algorithm~\ref{Alg:SG-Master} converges to a stationary point as long as the weighted sum rate over all subchannels monotonically increases as the iteration progresses.

Let $\{\bp^{(l)}, \bq^{(l)}, \bar{\bP}^{(l)}\}$ be the solution to Problem~$\Pteq$ obtained at the beginning of the $l$th iteration of Algorithm~\ref{Alg:SG-Master}, and $y(\bp^{(l)}, \bq^{(l)}, \bar{\bP}^{(l)})$ be the corresponding weighted sum rate over all subchannels.
In Algorithm~\ref{Alg:SG-Master}, two updates are performed in each iteration:
one for both the power allocation and subchannel assignment indicator vectors (i.e., $\bp$ and $\bq$),
and the other for the coupling vector (i.e., $\bar{\bP}$).
Accordingly, in the $l$th iteration, the weighted sum rate over all subchannels evolves as
\if\mydocumentclass0
	\begin{multline}
		y(\bp^{(l)}, \bq^{(l)}, \bar{\bP}^{(l)})
		\overset{(a)}{\to} y(\bp^{(l+1)}, \bq^{(l+1)}, \bar{\bP}^{(l)}) \\
		\overset{(b)}{\to} y(\bp^{(l+1)}, \bq^{(l+1)}, \bar{\bP}^{(l+1)}),
	\end{multline}
\else
	\begin{equation}
		y(\bp^{(l)}, \bq^{(l)}, \bar{\bP}^{(l)})
		\overset{(a)}{\to} y(\bp^{(l+1)}, \bq^{(l+1)}, \bar{\bP}^{(l)})
		\overset{(b)}{\to} y(\bp^{(l+1)}, \bq^{(l+1)}, \bar{\bP}^{(l+1)}),
	\end{equation}
\fi
where $(a)$ and $(b)$ are done by solving Problem~$\Slave$ for all $k\in\mathcal{K}$ and Problem~$\Master$, respectively.
Since Problem~$\Slave$ is to optimize $\{\bp_k,\bq_k\}$ so that the weighted sum rate on Subchannel~$k$ under the limited power of $\bar{P}_k^{(l)}$ is maximized, it is obvious that
\begin{equation} \label{Update by solving Slave}
	y(\bp^{(l)}, \bq^{(l)}, \bar{\bP}^{(l)}) \le y(\bp^{(l+1)}, \bq^{(l+1)}, \bar{\bP}^{(l)}).
\end{equation}
Similarly, since Problem~$\Master$ is to optimize $\bar{\bP}$ so that the weighted sum rate over all subchannels is maximized, it is obvious that
\begin{equation} \label{Update by solving Master}
	y(\bp^{(l+1)}, \bq^{(l+1)}, \bar{\bP}^{(l)}) \le y(\bp^{(l+1)}, \bq^{(l+1)}, \bar{\bP}^{(l+1)}).
\end{equation}
By combining \eqref{Update by solving Slave} and \eqref{Update by solving Master}, we can easily see that the weighted sum rate over all subchannels, i.e., the objective value of Problem~$\Pteq$, is monotonically increasing as the iteration progresses. \qed

\section{Proof of Theorem~\ref{thm:user_selection_exclusion}}
\label{prove:thm:user_selection_exclusion}
Let us denote the objective function of Problem~$\Slave$, its optimal solution, and its optimal value by $f_{\Slave}(\bp_k, \bq_k)$, $\{\bp_k^*, \bq_k^*\}$, and $f_{\Slave}^* = f_{\Slave}(\bp_k^*, \bq_k^*)$, respectively.
We also denote the objective function of Problem~$\Q$, its optimal solution, and its optimal value by $f_{\Q}(\bp_k)$, $\bp_k^\dagger$, and $f_{\Q}^\dagger = f_{\Q}(\bp_k^\dagger)$, respectively.
In addition, we define a feasible set for the subchannel assignment indicator vector for Subchannel~$k$ by $\mathcal{Q}_k=\{\bq_k\in\{0,1\}^N : \sum_{i\in\mathcal{N}} q_{k,i} \le M\}$, and denote a set of users accessing Subchannel~$k$ as $\mathcal{N}_k = \{i\in\mathcal{N}:q_{k,i}=1\}$.

Now, suppose that $\bq_k$ is arbitrarily given such that $\bq_k\in\mathcal{Q}_k$.
Then, Problem~$\Slave$ turns into
\begin{IEEEeqnarray}{c'l}\label{Prob:fhat}
	\IEEEyesnumber
	\maximize_{\bp_k} & \hat{f}_{\Slave}(\bp_k;\bq_k) \\
	\subjto & \sum_{i\in\mathcal{N}} p_{k,i} \le \bar{P}_k, \IEEEnonumber*\\
	& p_{k,i} \ge 0, ~ \forall i\in\mathcal{N},
\end{IEEEeqnarray}
where
\if\mydocumentclass0
	\begin{multline}
		\hat{f}_{\Slave}(\bp_k;\bq_k) \\
		= \sum_{i\in\mathcal{N}_k} \tilde{w}_i B_k \log_2 \left( 1 + \frac{p_{k,i}}{\sum_{j\in\mathcal{N}_k:\eta_{k,j}<\eta_{k,i}} p_{k,j} + \eta_{k,i}} \right).
	\end{multline}
\else
	\begin{equation}
		\hat{f}_{\Slave}(\bp_k;\bq_k)
		= \sum_{i\in\mathcal{N}_k} \tilde{w}_i B_k \log_2 \left( 1 + \frac{p_{k,i}}{\sum_{j\in\mathcal{N}_k:\eta_{k,j}<\eta_{k,i}} p_{k,j} + \eta_{k,i}} \right).
	\end{equation}
\fi
Here, it is worth noting that the problem in~\eqref{Prob:fhat} is identical with Problem~$\Q$ except that the user set is given as $\mathcal{N}_k$ instead of $\mathcal{N}$ in the objective function.
Since $\mathcal{N}_k$ is a subset of $\mathcal{N}$, it is obvious that $\hat{f}_{\Slave}^*(\bq_k) \le f_{\Q}^\dagger$ for any $\bq_k\in\mathcal{Q}_k$, where $\hat{f}_{\Slave}^*(\bq_k)$ is the optimal value of the problem in~\eqref{Prob:fhat} for a given $\bq_k$.
Thus, we can deduce that
\begin{equation} \label{eq:Slave le Q}
	f_{\Slave}^* \le f_{\Q}^\dagger.
\end{equation}

Next, suppose that the optimal solution, $\bp_k^\dagger$, to Problem~$\Q$ is given such that $\sum_{i\in\mathcal{N}} \mathbf{1}_{\{p_{k,i}^\dagger>0\}} \le M$, and let us define
\begin{equation}
	\hat{\bp}_k=\bp_k^\dagger ~ \textnormal{and} ~ \hat{\bq}_k=(\hat{q}_{k,i})_{\forall i\in\mathcal{N}},
\end{equation}
where $\hat{q}_{k,i}=\mathbf{1}_{\{p_{k,i}^\dagger>0\}}$ for all $i\in\mathcal{N}$.
Then, $\hat{\bp}_k$ and $\hat{\bq}_k$ satisfy the constraints in Problem~$\Slave$.
Also, by simple arithmetic operations, we can easily see that
\begin{equation} \label{eq:Slave opt = Q opt}
	f_{\Slave}(\hat{\bp}_k, \hat{\bq}_k) = f_{\Q}^\dagger.
\end{equation}
By~\eqref{eq:Slave le Q} and \eqref{eq:Slave opt = Q opt}, we can deduce that $\{\hat\bp_k, \hat\bq_k\}$ is an optimal solution to Problem~$\Slave$, i.e., $\bp_k^*=\hat\bp_k$ and $\bq_k^*=\hat\bq_k$. \qed

\section{Proof of Theorem~\ref{thm:2-user}}
\label{prove:thm:2-user}
Since the larger the transmission power, the higher the weighted sum rate, we can replace the first constraint in Problem~$\Qfour$ with $p_{k,\psi_k} + p_{k,\varphi_k} = \bar{P}_k$.
Then, by substituting $p_{k,\psi_k}$ for $\bar{P}_k-p_{k,\varphi_k}$, we can reformulate Problem~$\Qfour$ as
\begin{IEEEeqnarray}{c'l}
	\label{Prob:Appendix_One_Variable}
	\maximize_{0 \le p_{k,\varphi_k} \le \bar{P}^k} & g(p_{k,\varphi_k}),
\end{IEEEeqnarray}
where
\if\mydocumentclass0
	\begin{multline}
	g(p_{k,\varphi_k}) = \tilde{w}_{\psi_k} B_k \log_2 \left( \frac{\bar{P}_k + \eta_{k,\psi_k}}{p_{k,\varphi_k} + \eta_{k,\psi_k}} \right) \\
	+ \tilde{w}_{\varphi_k} B_k \log_2 \left( 1+ \frac{p_{k,\varphi_k}}{\eta_{k,\varphi_k}} \right).
	\end{multline}
\else
	\begin{equation}
	g(p_{k,\varphi_k}) = \tilde{w}_{\psi_k} B_k \log_2 \left( \frac{\bar{P}_k + \eta_{k,\psi_k}}{p_{k,\varphi_k} + \eta_{k,\psi_k}} \right)
	+ \tilde{w}_{\varphi_k} B_k \log_2 \left( 1+ \frac{p_{k,\varphi_k}}{\eta_{k,\varphi_k}} \right).
	\end{equation}
\fi
The derivative of $g(p_{k,\varphi_k})$ with respect to $p_{k,\varphi_k}$ is given as
\if\mydocumentclass0
	\begin{align}\label{eq:gprime}
		&g'(p_{k,\varphi_k}) \nonumber\\
		&= \frac{B_k}{\ln2} \times \left[ \frac{-\tilde{w}_{\psi_k}}{p_{k,\varphi_k} + \eta_{k,\psi_k}} + \frac{\tilde{w}_{\varphi_k}}{p_{k,\varphi_k} + \eta_{k,\varphi_k}} \right] \nonumber\\
		&= \frac{B_k}{\ln2} \times \frac{(\tilde{w}_{\varphi_k}-\tilde{w}_{\psi_k})p_{k,\varphi_k} + \tilde{w}_{\varphi_k}\eta_{k,\psi_k} - \tilde{w}_{\psi_k}\eta_{k,\varphi_k}}{(p_{k,\varphi_k} + \eta_{k,\psi_k})(p_{k,\varphi_k} + \eta_{k,\varphi_k})}.
	\end{align}
\else
	\begin{align}\label{eq:gprime}
		g'(p_{k,\varphi_k})
		&= \frac{B_k}{\ln2} \times \left[ \frac{-\tilde{w}_{\psi_k}}{p_{k,\varphi_k} + \eta_{k,\psi_k}} + \frac{\tilde{w}_{\varphi_k}}{p_{k,\varphi_k} + \eta_{k,\varphi_k}} \right] \nonumber\\
		&= \frac{B_k}{\ln2} \times \frac{(\tilde{w}_{\varphi_k}-\tilde{w}_{\psi_k})p_{k,\varphi_k} + \tilde{w}_{\varphi_k}\eta_{k,\psi_k} - \tilde{w}_{\psi_k}\eta_{k,\varphi_k}}{(p_{k,\varphi_k} + \eta_{k,\psi_k})(p_{k,\varphi_k} + \eta_{k,\varphi_k})}.
	\end{align}
\fi
From ~\eqref{eq:gprime}, we have $g'(\hat{p}_{k,\varphi_k}) = 0$ if and only if
\begin{equation}\label{eq:p2hat}
	\hat{p}_{k,\varphi_k} = \frac{\tilde{w}_{\psi_k}\eta_{k,\varphi_k}-\tilde{w}_{\varphi_k}\eta_{k,\psi_k}}{\tilde{w}_{\varphi_k}-\tilde{w}_{\psi_k}}.
\end{equation}
Using~\eqref{eq:gprime} and \eqref{eq:p2hat}, we can derive an optimal solution, $p_{k,\varphi_k}^\star$, to the problem in~\eqref{Prob:Appendix_One_Variable} by considering the following three mutually exclusive cases.
\begin{enumerate}
\item 	Suppose that $\tilde{w}_{\varphi_k}/\tilde{w}_{\psi_k} > 1$.
		Then, according to~\eqref{eq:gprime}, $g'(p_{k,\varphi_k}) < 0$ if $p_{k,\varphi_k} < \hat{p}_{k,\varphi_k}$ and $g'(p_{k,\varphi_k}) \ge 0$ otherwise.
		Also, according to~\eqref{eq:p2hat}, $\hat{p}_{k,\varphi_k} < 0$ since $\eta_{k,\varphi_k} < \eta_{k,\psi_k}$.
		Thus, $g(p_{k,\varphi_k})$ is an increasing function on $[0, \bar{P}_k]$, resulting in $p_{k,\varphi_k}^\star=\bar{P}_k$.
\item	Suppose that $\tilde{w}_{\varphi_k}/\tilde{w}_{\psi_k} = 1$.
		Then, according to~\eqref{eq:gprime}, $g'(p_{k,\varphi_k}) > 0$ for any $p_{k,\varphi_k} \in [0, \bar{P}_k]$.
		Thus, $g(p_{k,\varphi_k})$ is an increasing function on $[0, \bar{P}_k]$, resulting in $p_{k,\varphi_k}^\star=\bar{P}_k$.
\item	Suppose that $\tilde{w}_{\varphi_k}/\tilde{w}_{\psi_k} < 1$.
		Then, according to~\eqref{eq:gprime}, $g'(p_{k,\varphi_k}) > 0$ if $p_{k,\varphi_k} < \hat{p}_{k,\varphi_k}$ and $g'(p_{k,\varphi_k}) \le 0$ otherwise.
		Also, according to~\eqref{eq:p2hat}, we have the following two inequalities.
		\begin{align}
			\hat{p}_{k,\varphi_k} \le 0 &~\Leftrightarrow~ \frac{\tilde{w}_{\varphi_k}}{\tilde{w}_{\psi_k}} \le \frac{\eta_{k,\varphi_k}}{\eta_{k,\psi_k}},\label{eq:cond1}\\
			\hat{p}_{k,\varphi_k} > \bar{P}_k &~\Leftrightarrow~ \frac{\tilde{w}_{\varphi_k}}{\tilde{w}_{\psi_k}} > \frac{\bar{P}_k+\eta_{k,\varphi_k}}{\bar{P}_k+\eta_{k,\psi_k}}.\label{eq:cond2}
		\end{align}
		Consequently, if \eqref{eq:cond1} is met, $g(p_{k,\varphi_k})$ is a decreasing function on $[0, \bar{P}_k]$, resulting in $p_{k,\varphi_k}^\star=0$; if \eqref{eq:cond2} is met, $g(p_{k,\varphi_k})$ is an increasing function on $[0, \bar{P}_k]$, resulting in $p_{k,\varphi_k}^\star=\bar{P}_k$; and if neither \eqref{eq:cond1} nor \eqref{eq:cond2} is met, $g(p_{k,\varphi_k})$ is an increasing function on $[0, \hat{p}_{k,\varphi_k}]$ but a decreasing function on $[\hat{p}_{k,\varphi_k}, \bar{P}_k]$, resulting in $p_{k,\varphi_k}^\star=\hat{p}_{k,\varphi_k}$.
\end{enumerate}
%
%
%
Note that since $\eta_{k,\varphi_k}<\eta_{k,\psi_k}$, the right-hand side of \eqref{eq:cond2} is always less than one, i.e., $(\bar{P}_k+\eta_{k,\varphi_k})/(\bar{P}_k+\eta_{k,\psi^k})<1$.
Hence, the results for the three cases can be combined as
\begin{equation}
	p_{k,\varphi_k}^\star =
	\begin{dcases}
		0,				& \textnormal{if } {\tilde{w}_{\varphi_k}}/{\tilde{w}_{\psi_k}} \le \eta_{k,\varphi_k}/\eta_{k,\psi_k}, \\
		\bar{P}_k,		& \textnormal{if } {\tilde{w}_{\varphi_k}}/{\tilde{w}_{\psi_k}} > (\bar{P}_k+\eta_{k,\varphi_k})/(\bar{P}_k+\eta_{k,\psi_k}), \\
		\hat{p}_{k,\varphi_k}	& \textnormal{otherwise.}
	\end{dcases}
\end{equation}
Lastly, $p_{k,\psi_k}^\star = \bar{P}_k - p_{k,\varphi_k}^\star$ since $p_{k,\psi_k} + p_{k,\varphi_k} = \bar{P}_k$. \qed

\section{Proof of the equivalence of \eqref{eq:phik2} and \eqref{eq:phik3}}
\label{prove:phi_k_equivalence}
To show the equivalence of \eqref{eq:phik2} and \eqref{eq:phik3}, it is sufficient to show that
\begin{equation}\label{eq:Pcond1}
	\frac{\tilde{w}_{\varphi_k}}{\tilde{w}_{\psi_k}} \le C_k^2(\bar{P}_k) ~\Leftrightarrow~ \frac{\tilde{w}_{\varphi_k}}{\tilde{w}_{\psi_k}} < 1 ~ \textnormal{and} ~ \bar{P}_k \ge C_k^4.
\end{equation}
First of all, the fact that the left-hand side holds if the right-hand side holds can be easily proved by simple arithmetic operations.
Hence, we only prove that if the left-hand side holds, then its right-hand side holds.
To this end, we first have
\if\mydocumentclass0
	\begin{align}
		&\frac{\tilde{w}_{\varphi_k}}{\tilde{w}_{\psi_k}} \le C_k^2(\bar{P}_k) \label{eq:Pcond2}\\
		\overset{(a)}{\Leftrightarrow}~ &\frac{\tilde{w}_{\varphi_k}}{\tilde{w}_{\psi_k}} \le \frac{\bar{P}_k + \eta_{k,\varphi_k}}{\bar{P}_k + \eta_{k,\psi_k}} \\
		\overset{(b)}{\Leftrightarrow}~ &\bar{P}_k (\tilde{w}_{\varphi_k}-\tilde{w}_{\psi_k}) \le \tilde{w}_{\psi_k} \eta_{k,\varphi_k} - \tilde{w}_{\varphi_k} \eta_{k,\psi_k}. \label{eq:Pcond3}
	\end{align}
\else
	\begin{align}
		\frac{\tilde{w}_{\varphi_k}}{\tilde{w}_{\psi_k}} \le C_k^2(\bar{P}_k)
		&~\overset{(a)}{\Leftrightarrow}~ \frac{\tilde{w}_{\varphi_k}}{\tilde{w}_{\psi_k}} \le \frac{\bar{P}_k + \eta_{k,\varphi_k}}{\bar{P}_k + \eta_{k,\psi_k}} \label{eq:Pcond2}\\
		&~\overset{(b)}{\Leftrightarrow}~ \bar{P}_k (\tilde{w}_{\varphi_k}-\tilde{w}_{\psi_k}) \le \tilde{w}_{\psi_k} \eta_{k,\varphi_k} - \tilde{w}_{\varphi_k} \eta_{k,\psi_k}. \label{eq:Pcond3}
	\end{align}
\fi
where $(a)$ is done by the definition of $C_k^2(\bar{P}_k)$ in~\eqref{eq:C1 and C2}, and $(b)$ is done by simple arithmetic operations.
Then, we can deduce that if \eqref{eq:Pcond2} holds, then
\begin{equation}
	\frac{\tilde{w}_{\varphi_k}}{\tilde{w}_{\psi_k}} < 1 ~ \textnormal{and} ~ \bar{P}_k \ge C_k^4,
\end{equation}
where the first inequality is derived from the fact that $C_k^2(\bar{P}_k)$ is always less than $1$ since $\eta_{k,\psi_k} > \eta_{k,\varphi_k}$, and the second inequality is obtained by dividing by $\tilde{w}_{\varphi_k}-\tilde{w}_{\psi_k}$ on both sides of \eqref{eq:Pcond3}.
As a result, we can conclude that \eqref{eq:Pcond1} holds.

\section{Proof of Proposition~\ref{thm:convexity_of_phi}}
\label{prove:thm:convexity_of_phi}

First, consider the case where
\begin{equation} \label{prove:thm:convexity_of_phi:eq:cond1}
	\tilde{w}_{\varphi_k}/\tilde{w}_{\psi_k} \ge 1.
\end{equation}
Then, according to~\eqref{eq:phik3}, we have $\phi_k^*(\bar{P}_k) = \tilde{w}_{\varphi_k} B_k \log_2 ( 1 + \bar{P}_k/\eta_{k,\varphi_k} )$, which is a logarithmic function that can be defined on $[0,\infty)$.
Hence, $\phi_k^*$ is a continuously differentiable concave function of $\bar{P}_k$ on $[0, \infty)$ if \eqref{prove:thm:convexity_of_phi:eq:cond1} holds.

Next, consider the other case where
\begin{equation} \label{prove:thm:convexity_of_phi:eq:cond2}
	\tilde{w}_{\varphi_k}/\tilde{w}_{\psi_k} < 1.
\end{equation}
Then, according to~\eqref{eq:phik3}, we have
\begin{equation} \label{prove:thm:convexity_of_phi:eq:phik_piecewise}
	\phi_k^*(\bar{P}_k) = \begin{dcases}
		\phi_{k,1}(\bar{P}_k), & \textnormal{if } \bar{P}_k \ge C_k^4, \\
		\phi_{k,2}(\bar{P}_k), & \textnormal{otherwise,}
	\end{dcases}
\end{equation}
where
\begin{align}
	\phi_{k,1}(\bar{P}_k) &= \tilde{w}_{\psi_k} B_k \log_2 \left( 1 + \frac{\bar{P}_k}{\eta_{k,\psi_k}} \right) + C_k^3, \label{prove:thm:convexity_of_phi:eq:phi_k1}\\
	\phi_{k,2}(\bar{P}_k) &= \tilde{w}_{\varphi_k} B_k \log_2 \left( 1 + \frac{\bar{P}_k}{\eta_{k,\varphi_k}} \right). \label{prove:thm:convexity_of_phi:eq:phi_k2}
\end{align}
Both $\phi_{k,1}$ and $\phi_{k,2}$ are logarithmic functions that can be defined on $[0,\infty)$.
Hence, each of them is a continuously differentiable concave function of $\bar{P}_k$ on $[0, \infty)$.
However, as shown in~\eqref{prove:thm:convexity_of_phi:eq:phik_piecewise}, $\phi_k^*$ is a piecewise function of $\bar{P}_k$ with a breakpoint at $\bar{P}_k = C_k^4$.
Hence, to prove Proposition~\ref{thm:convexity_of_phi}, we only need to show that $\phi_{k,1}(\bar{P}_k)=\phi_{k,2}(\bar{P}_k)$ and $\phi_{k,1}'(\bar{P}_k)=\phi_{k,2}'(\bar{P}_k)$ at the breakpoint, where $\phi_{k,1}'$ and $\phi_{k,2}'$ denote the derivative functions of $\phi_{k,1}$ and $\phi_{k,2}$, respectively, with respect to $\bar{P}_k$.

We first show that $\phi_{k,1}(\bar{P}_k)=\phi_{k,2}(\bar{P}_k)$ at $\bar{P}_k = C_k^4$.
From the definitions of $C_k^3$ and $C_k^4$ given in \eqref{eq:C3} and \eqref{eq:C4}, respectively, we have
\if\mydocumentclass0
	\begin{align} \label{prove:thm:convexity_of_phi:eq:phik1 of Ck4}
		\phi_{k,1}(C_k^4)
		&= \begin{multlined}[t]
			\tilde{w}_{\psi_k} B_k \log_2 \left( 1 + \frac{\frac{\tilde{w}_{\psi_k} \eta_{k,\varphi_k} - \tilde{w}_{\varphi_k} \eta_{k,\psi_k}}{\tilde{w}_{\varphi_k}-\tilde{w}_{\psi_k}}}{\eta_{k,\psi_k}} \right) \\
			+ \tilde{w}_{\psi_k} B_k \log_2 \left( \frac{\tilde{w}_{\varphi_k} - \tilde{w}_{\psi_k}}{\eta_{k,\varphi_k}-\eta_{k,\psi_k}} \cdot \frac{\eta_{k,\psi_k}}{\tilde{w}_{\psi_k}} \right) \\
			+ \tilde{w}_{\varphi_k} B_k \log_2 \left( \frac{\eta_{k,\varphi_k}-\eta_{k,\psi_k}}{\tilde{w}_{\varphi_k}-\tilde{w}_{\psi_k}} \cdot \frac{\tilde{w}_{\varphi_k}}{\eta_{k,\varphi_k}} \right)
		\end{multlined} \nonumber\\
		&= \tilde{w}_{\varphi_k} B_k \log_2 \left( \frac{\eta_{k,\varphi_k}-\eta_{k,\psi_k}}{\tilde{w}_{\varphi_k}-\tilde{w}_{\psi_k}} \cdot \frac{\tilde{w}_{\varphi_k}}{\eta_{k,\varphi_k}} \right) \nonumber\\
		&= \tilde{w}_{\varphi_k} B_k \log_2 \left( 1 + \frac{\frac{\tilde{w}_{\psi_k}\eta_{k,\varphi_k}-\tilde{w}_{\varphi_k}\eta_{k,\psi_k}}{\tilde{w}_{\varphi_k}-\tilde{w}_{\psi_k}}}{\eta_{k,\varphi_k}} \right) \nonumber\\
		&= \phi_{k,2}(C_k^4).
	\end{align}
\else
	\begin{align} \label{prove:thm:convexity_of_phi:eq:phik1 of Ck4}
		\phi_{k,1}(C_k^4)
		&= \begin{multlined}[t]
			\tilde{w}_{\psi_k} B_k \log_2 \left( 1 + \frac{\frac{\tilde{w}_{\psi_k} \eta_{k,\varphi_k} - \tilde{w}_{\varphi_k} \eta_{k,\psi_k}}{\tilde{w}_{\varphi_k}-\tilde{w}_{\psi_k}}}{\eta_{k,\psi_k}} \right) \\
			+ \tilde{w}_{\psi_k} B_k \log_2 \left( \frac{\tilde{w}_{\varphi_k} - \tilde{w}_{\psi_k}}{\eta_{k,\varphi_k}-\eta_{k,\psi_k}} \cdot \frac{\eta_{k,\psi_k}}{\tilde{w}_{\psi_k}} \right)
			+ \tilde{w}_{\varphi_k} B_k \log_2 \left( \frac{\eta_{k,\varphi_k}-\eta_{k,\psi_k}}{\tilde{w}_{\varphi_k}-\tilde{w}_{\psi_k}} \cdot \frac{\tilde{w}_{\varphi_k}}{\eta_{k,\varphi_k}} \right)
		\end{multlined} \nonumber\\
		&= \tilde{w}_{\varphi_k} B_k \log_2 \left( \frac{\eta_{k,\varphi_k}-\eta_{k,\psi_k}}{\tilde{w}_{\varphi_k}-\tilde{w}_{\psi_k}} \cdot \frac{\tilde{w}_{\varphi_k}}{\eta_{k,\varphi_k}} \right) \nonumber\\
		&= \tilde{w}_{\varphi_k} B_k \log_2 \left( 1 + \frac{\frac{\tilde{w}_{\psi_k}\eta_{k,\varphi_k}-\tilde{w}_{\varphi_k}\eta_{k,\psi_k}}{\tilde{w}_{\varphi_k}-\tilde{w}_{\psi_k}}}{\eta_{k,\varphi_k}} \right) \nonumber\\
		&= \phi_{k,2}(C_k^4).
	\end{align}
\fi

We now show that $\phi_{k,1}'(\bar{P}_k)=\phi_{k,2}'(\bar{P}_k)$ at $\bar{P}_k = C_k^4$.
The derivative of $\phi_{k,1}$ with respect to $\bar{P}_k$ is given by
\begin{equation}
	\phi_{k,1}'(\bar{P}_k) = \frac{\tilde{w}_{\psi_k} B_k}{\ln2} \cdot \frac{1}{\bar{P}_k + \eta_{k,\psi_k}}.
\end{equation}
Hence, we have
\begin{align} \label{prove:thm:convexity_of_phi:eq:phik1' of Ck4}
	\phi_{k,1}'(C_k^4) &= \frac{\tilde{w}_{\psi_k} B_k}{\ln2} \cdot \frac{1}{\frac{\tilde{w}_{\psi_k}\eta_{k,\varphi_k}-\tilde{w}_{\varphi_k}\eta_{k,\psi_k}}{\tilde{w}_{\varphi_k}-\tilde{w}_{\psi_k}} + \eta_{k,\psi_k}} \nonumber\\
	&= \frac{B_k}{\ln2} \cdot \frac{\tilde{w}_{\varphi_k}-\tilde{w}_{\psi_k}}{\eta_{k,\varphi_k}-\eta_{k,\psi_k}}
\end{align}
Next, the derivative of $\phi_{k,2}$ with respect to $\bar{P}_k$ is given by
\begin{equation}
	\phi_{k,2}'(\bar{P}_k) = \frac{\tilde{w}_{\varphi_k} B_k}{\ln2} \cdot \frac{1}{\bar{P}_k + \eta_{k,\varphi_k}},
\end{equation}
and thus,
\begin{align} \label{prove:thm:convexity_of_phi:eq:phik1p of Ck4}
	\phi_{k,2}'(C_k^4)
	&= \frac{\tilde{w}_{\varphi_k} B_k}{\ln2} \cdot \frac{1}{\frac{\tilde{w}_{\psi_k}\eta_{k,\varphi_k}-\tilde{w}_{\varphi_k}\eta_{k,\psi_k}}{\tilde{w}_{\varphi_k}-\tilde{w}_{\psi_k}} + \eta_{k,\varphi_k}} \nonumber\\
	&= \frac{B_k}{\ln2} \cdot \frac{\tilde{w}_{\varphi_k}-\tilde{w}_{\psi_k}}{\eta_{k,\varphi_k}-\eta_{k,\psi_k}} \nonumber\\
	&= \phi_{k,1}'(C_k^4).
\end{align}
Hence, by \eqref{prove:thm:convexity_of_phi:eq:phik1 of Ck4} and \eqref{prove:thm:convexity_of_phi:eq:phik1p of Ck4}, $\phi_k^*$ is a continuously differentiable concave function of $\bar{P}_k$ on $[0, \infty)$ if \eqref{prove:thm:convexity_of_phi:eq:cond2} holds. \qed

%
%

\section{Proof of Theorem~\ref{thm:Pkstar}}
\label{prove:thm:Pkstar}

We first introduce two indicator variables $e_k^1$ and $e_k^2$, where $e_k^1$ is equal to one if $\tilde{w}_{\varphi_k}/\tilde{w}_{\psi_k}<1$ and $\bar{P}_k \ge C_k^4$ and zero otherwise, and $e_k^2$ is equal to one if $\tilde{w}_{\varphi_k}/\tilde{w}_{\psi_k} \ge 1$ or $\bar{P}_k < C_k^4$ and zero otherwise.
Note that since the indicators are mutually exclusive, $e_k^1+e_k^2$ is always one.
Then, $\phi_k^*(\bar{P}_k)$ in \eqref{eq:phik3} can be expressed as
\if\mydocumentclass0
	\begin{multline}
		\phi_k^*(\bar{P}_k) = e_k^1 \cdot \left(\tilde{w}_{\psi_k} B_k \log_2 \left( 1 + \frac{\bar{P}_k}{\eta_{k,\psi_k}} \right) + C_k^3\right) \\
		+ e_k^2 \cdot \tilde{w}_{\varphi_k} B_k \log_2 \left( 1 + \frac{\bar{P}_k}{\eta_{k,\varphi_k}} \right).
	\end{multline}
\else
	\begin{equation}
		\phi_k^*(\bar{P}_k) = e_k^1 \cdot \left(\tilde{w}_{\psi_k} B_k \log_2 \left( 1 + \frac{\bar{P}_k}{\eta_{k,\psi_k}} \right) + C_k^3\right) 
		+ e_k^2 \cdot \tilde{w}_{\varphi_k} B_k \log_2 \left( 1 + \frac{\bar{P}_k}{\eta_{k,\varphi_k}} \right).
	\end{equation}
\fi
In addition, since the larger the total transmission power budget, the higher the objective value can be achieved, we can replace the first inequality constraint in Problem~$\Master$ with $\sum_{k\in\mathcal{K}} \bar{P}_k = \Pmax$.

Now, let $\bar{\bP}^*\in\mathbb{R}^K$ be an optimal solution to problem~$\Master$, $\underline{\bm{\muup}}^*\in\mathbb{R}^K$ and $\bar{\bm{\muup}}^*\in\mathbb{R}^K$ be Lagrange multiplier vectors for the inequality constraints $\bar{\bP} \succeq \mathbf{0}$ and $\bar{\bP} \preceq \bar{\bP}_{\textnormal{max}}$, respectively, where $\mathbf{0}$ is the $K$-dimensional zero vector and $\bar{\bP}_{\textnormal{max}}=(\Pmaxk)_{\forall k\in\mathcal{K}}$, and $\nu^*\in\mathbb{R}$ be a multiplier for the equality constraint $\sum_{k\in\mathcal{K}} \bar{P}_k = \Pmax$.
Then, the KKT conditions are obtained as
\if\mydocumentclass0
	\begin{align}
		&\begin{aligned}[b]
			\frac{e_k^1}{\ln2} \cdot \frac{\tilde{w}_{\psi_k} B_k}{\bar{P}_k^* + \eta_{k,\psi_k}} &+ \frac{e_k^2}{\ln2} \cdot \frac{\tilde{w}_{\varphi_k} B_k}{\bar{P}_k^* + \eta_{k,\varphi_k}} \\
			&+ \underline{\mu}_k^* - \bar{\mu}_k^* - \nu^* = 0, ~ \forall k\in\mathcal{K},
		\end{aligned} \label{KKT:Stationarity}\\
		&\sum_{k\in\mathcal{K}} \bar{P}_k^* = \Pmax, ~ 0 \le \bar{P}_k^* \le \Pmaxk, ~ \forall k\in\mathcal{K}, \label{KKT:Primal feasibility}\\
		&\underline{\mu}_k^* \ge 0, ~ \underline{\mu}_k^* \bar{P}_k^* = 0, ~ \forall k\in\mathcal{K}, \label{KKT:Dual feasibility and Complementary slackness for underline mu}\\
		&\bar{\mu}_k^* \ge 0, ~ \bar{\mu}_k^* (\bar{P}_k^* - \Pmaxk) = 0, ~ \forall k\in\mathcal{K}. \label{KKT:Dual feasibility and Complementary slackness for bar mu}
	\end{align}
\else
	\begin{align}
		&\frac{e_k^1}{\ln2} \cdot \frac{\tilde{w}_{\psi_k} B_k}{\bar{P}_k^* + \eta_{k,\psi_k}} + \frac{e_k^2}{\ln2} \cdot \frac{\tilde{w}_{\varphi_k} B_k}{\bar{P}_k^* + \eta_{k,\varphi_k}} + \underline{\mu}_k^* - \bar{\mu}_k^* - \nu^* = 0, ~ \forall k\in\mathcal{K}, \label{KKT:Stationarity}\\
		&\sum_{k\in\mathcal{K}} \bar{P}_k^* = \Pmax, ~ 0 \le \bar{P}_k^* \le \Pmaxk, ~ \forall k\in\mathcal{K}, \label{KKT:Primal feasibility}\\
		&\underline{\mu}_k^* \ge 0, ~ \underline{\mu}_k^* \bar{P}_k^* = 0, ~ \forall k\in\mathcal{K}, \label{KKT:Dual feasibility and Complementary slackness for underline mu}\\
		&\bar{\mu}_k^* \ge 0, ~ \bar{\mu}_k^* (\bar{P}_k^* - \Pmaxk) = 0, ~ \forall k\in\mathcal{K}. \label{KKT:Dual feasibility and Complementary slackness for bar mu}
	\end{align}
\fi
We directly solve these equations to find $\bar{\bP}^*$, $\underline{\bm{\muup}}^*$, $\bar{\bm{\muup}}^*$, and $\nu^*$.
To this end, we first eliminate $\underline{\bm{\muup}}^*$ by rearranging \eqref{KKT:Stationarity} for $\underline{\mu}_k^*$ and then plugging it into \eqref{KKT:Dual feasibility and Complementary slackness for underline mu}.
Thereby, we have
\begin{gather}
	\bar{\mu}_k^* + \nu^* - \left( \frac{e_k^1}{\ln2} \cdot \frac{\tilde{w}_{\psi_k} B_k}{\bar{P}_k^* + \eta_{k,\psi_k}} + \frac{e_k^2}{\ln2} \cdot \frac{\tilde{w}_{\varphi_k} B_k}{\bar{P}_k^* + \eta_{k,\varphi_k}} \right) \ge 0, \label{KKT2:Dual feasibility for bar mu and nu} \\
	\left( \bar{\mu}_k^* + \nu^* - \left( \frac{e_k^1}{\ln2} \cdot \frac{\tilde{w}_{\psi_k} B_k}{\bar{P}_k^* + \eta_{k,\psi_k}} + \frac{e_k^2}{\ln2} \cdot \frac{\tilde{w}_{\varphi_k} B_k}{\bar{P}_k^* + \eta_{k,\varphi_k}} \right) \right) \times \bar{P}_k^* = 0. \label{KKT2:Complementray slackness for bar mu and nu}
\end{gather}

We now find $\bar{\bP}^*$, $\bar{\bm{\muup}}^*$, and $\nu^*$ with \eqref{KKT:Primal feasibility}, \eqref{KKT:Dual feasibility and Complementary slackness for bar mu}, \eqref{KKT2:Dual feasibility for bar mu and nu}, and \eqref{KKT2:Complementray slackness for bar mu and nu}.
First, suppose the case where
\begin{equation} \label{eq:suppose1}
	\bar{\mu}_k^* + \nu^* < \frac{e_k^1}{\ln2} \cdot \frac{\tilde{w}_{\psi_k} B_k}{\eta_{k,\psi_k}} + \frac{e_k^2}{\ln2} \cdot \frac{\tilde{w}_{\varphi_k} B_k}{\eta_{k,\varphi_k}}.
\end{equation}
Then, \eqref{KKT2:Dual feasibility for bar mu and nu} can hold only if $\bar{P}_k^* > 0$.
Thus, by~\eqref{KKT2:Complementray slackness for bar mu and nu}, we have
\begin{equation} \label{eq:Pbark for suppose1}
	\bar{\mu}_k^* + \nu^* = \frac{e_k^1}{\ln2} \cdot \frac{\tilde{w}_{\psi_k} B_k}{\bar{P}_k^* + \eta_{k,\psi_k}} + \frac{e_k^2}{\ln2} \cdot \frac{\tilde{w}_{\varphi_k} B_k}{\bar{P}_k^* + \eta_{k,\varphi_k}}.
\end{equation}
Rearranging \eqref{eq:Pbark for suppose1} for $\bar{P}_k^*$, we can conclude that if \eqref{eq:suppose1} holds,
\if\mydocumentclass0
	\begin{multline}
		\bar{P}_k^* = e_k^1 \cdot \left( \frac{\tilde{w}_{\psi_k} B_k}{(\bar{\mu}_k^*+\nu^*)\ln2} - \eta_{k,\psi_k} \right) \\
		+ e_k^2 \cdot \left( \frac{\tilde{w}_{\varphi_k} B_k}{(\bar{\mu}_k^*+\nu^*)\ln2} - \eta_{k,\varphi_k} \right).
	\end{multline}
\else
	\begin{equation}
		\bar{P}_k^* = e_k^1 \cdot \left( \frac{\tilde{w}_{\psi_k} B_k}{(\bar{\mu}_k^*+\nu^*)\ln2} - \eta_{k,\psi_k} \right)
		+ e_k^2 \cdot \left( \frac{\tilde{w}_{\varphi_k} B_k}{(\bar{\mu}_k^*+\nu^*)\ln2} - \eta_{k,\varphi_k} \right).
	\end{equation}
\fi
Now, suppose the opposite case where
\begin{equation} \label{eq:suppose2}
	\bar{\mu}_k^* + \nu^* \ge \frac{e_k^1}{\ln2} \cdot \frac{\tilde{w}_{\psi_k} B_k}{\eta_{k,\psi_k}} + \frac{e_k^2}{\ln2} \cdot \frac{\tilde{w}_{\varphi_k} B_k}{\eta_{k,\varphi_k}}.
\end{equation}
In this case, $\bar{P}_k^*>0$ is impossible since it implies that
\begin{align}
	\bar{\mu}_k^* + \nu^*
	&\ge \frac{e_k^1}{\ln2} \cdot \frac{\tilde{w}_{\psi_k} B_k}{\eta_{k,\psi_k}} + \frac{e_k^2}{\ln2} \cdot \frac{\tilde{w}_{\varphi_k} B_k}{\eta_{k,\varphi_k}} \nonumber \\
	&> \frac{e_k^1}{\ln2} \cdot \frac{\tilde{w}_{\psi_k} B_k}{\bar{P}_k^*+\eta_{k,\psi_k}} + \frac{e_k^2}{\ln2} \cdot \frac{\tilde{w}_{\varphi_k} B_k}{\bar{P}_k^*+\eta_{k,\varphi_k}},
\end{align}
which violates~\eqref{KKT2:Complementray slackness for bar mu and nu}.
Hence, we can conclude that $\bar{P}_k^*=0$ if \eqref{eq:suppose2} holds.
Combining the results for the above two cases, we can simply express $\bar{P}_k^*$ as
\if\mydocumentclass0
	\begin{multline} \label{eq:Pk solution ge 0}
		\bar{P}_k^* = \max\Biggl\{ 0, \, e_k^1 \cdot \left( \frac{\tilde{w}_{\psi_k} B_k}{(\bar{\mu}_k^*+\nu^*)\ln2} - \eta_{k,\psi_k} \right) \\
		+ e_k^2 \cdot \left( \frac{\tilde{w}_{\varphi_k} B_k}{(\bar{\mu}_k^*+\nu^*)\ln2} - \eta_{k,\varphi_k} \right) \Biggr\}.
	\end{multline}
\else
	\begin{equation} \label{eq:Pk solution ge 0}
		\bar{P}_k^* = \max\Biggl\{ 0, \, e_k^1 \cdot \left( \frac{\tilde{w}_{\psi_k} B_k}{(\bar{\mu}_k^*+\nu^*)\ln2} - \eta_{k,\psi_k} \right)
		+ e_k^2 \cdot \left( \frac{\tilde{w}_{\varphi_k} B_k}{(\bar{\mu}_k^*+\nu^*)\ln2} - \eta_{k,\varphi_k} \right) \Biggr\}.
	\end{equation}
\fi
In addition, by the complementary slackness conditions for~$\bar{\bm{\muup}}^*$ in~\eqref{KKT:Dual feasibility and Complementary slackness for bar mu}, we can see that if $\bar{P}_k^* < \Pmaxk$, $\bar{\mu}_k^*$ should be zero; otherwise, $\bar{\mu}_k^*$ should be a certain value such that $\bar{P}_k^*=\Pmaxk$.
According to these facts, the Lagrange multiplier~$\bar{\mu}_k^*$ can be eliminated in~\eqref{eq:Pk solution ge 0}, leaving
\begin{equation} \label{eq:Pk solution with condition for Pk}
	\bar{P}_k^* = \Biggl[ e_k^1 \cdot \left( \mu^* \tilde{w}_{\psi_k} B_k - \eta_{k,\psi_k} \right) + e_k^2 \cdot \left( \mu^* \tilde{w}_{\varphi_k} B_k - \eta_{k,\varphi_k} \right) \Biggr]_0^\Pmaxk,
\end{equation}
where $\mu^* = 1/(\nu^*\ln2)$.
For simple notation, let $f_{k,\psi_k}(\mu^*) = \mu^* \tilde{w}_{\psi_k} B_k - \eta_{k,\psi_k}$ and $f_{k,\varphi_k}(\mu^*) = \mu^* \tilde{w}_{\varphi_k} B_k - \eta_{k,\varphi_k}$.
Then, by the definitions of $e_k^1$ and $e_k^2$, \eqref{eq:Pk solution with condition for Pk} can be expressed as
\begin{equation} \label{eq:Pkstar with fkmu}
	\bar{P}_k^* = \left[f_k(\mu^*)\right]_0^\Pmaxk,
\end{equation}
where
\begin{equation} \label{eq:fkmu}
	f_k(\mu^*) = \begin{dcases}
		f_{k,\psi_k}(\mu^*), & \textnormal{if } \tilde{w}_{\varphi_k}/\tilde{w}_{\psi_k} < 1 \textnormal{ and } f_k(\mu^*) \ge C_k^4, \\
		f_{k,\varphi_k}(\mu^*), & \textnormal{if } \tilde{w}_{\varphi_k}/\tilde{w}_{\psi_k} < 1 \textnormal{ and } f_k(\mu^*) < C_k^4, \\
		f_{k,\varphi_k}(\mu^*), & \textnormal{if } \tilde{w}_{\varphi_k}/\tilde{w}_{\psi_k} \ge 1. \\
	\end{dcases}
\end{equation}
In the case where $\tilde{w}_{\varphi_k}/\tilde{w}_{\psi_k} < 1$, $f_k(\mu^*)$ is piecewise linear with a breakpoint at $f_k(\mu^*) = C_k^4$, but continuous since both $f_{k,\psi_k}(\mu^*)$ and $f_{k,\varphi_k}(\mu^*)$ have the same value $C_k^4$ at $\mu^*=C_k^5$.
Since both $f_{k,\psi_k}(\mu^*)$ and $f_{k,\varphi_k}(\mu^*)$ are monotonically increasing, \eqref{eq:fkmu} can be rewritten as
\begin{equation} \label{eq:fkmu2}
	f_k(\mu^*) = \begin{dcases}
		f_{k,\psi_k}(\mu^*), & \textnormal{if } \tilde{w}_{\varphi_k}/\tilde{w}_{\psi_k} < 1 \textnormal{ and } \mu^* \ge C_k^5, \\
		f_{k,\varphi_k}(\mu^*), & \textnormal{if } \tilde{w}_{\varphi_k}/\tilde{w}_{\psi_k} < 1 \textnormal{ and } \mu^* < C_k^5, \\
		f_{k,\varphi_k}(\mu^*), & \textnormal{if } \tilde{w}_{\varphi_k}/\tilde{w}_{\psi_k} \ge 1. \\
	\end{dcases}
\end{equation}
By plugging \eqref{eq:fkmu2} into \eqref{eq:Pkstar with fkmu}, we finally have \eqref{Pkstar in thm}. \qed

\section{Proof of Theorem~\ref{thm:zero-duality-gap}}
\label{prove:thm:zero-duality-gap}

In order to prove the strong duality between Problem~$\Ptwo$ and its dual problem, Problem~$\Dual$, we utilize the \textit{time-sharing} condition proposed in~\cite{yu2006dual}, which is defined as follows.
\begin{definition}
Let $\{\barbp_x, \barbq_x\}$ and $\{\barbp_y, \barbq_y\}$ be optimal solutions to Problem~$\Ptwo$ with $\bar{\bR}_\textnormal{min}=\bar{\bR}_x$ and $\bar{\bR}_\textnormal{min}=\bar{\bR}_y$, respectively, where $\bar{\bR}_{\textnormal{min}}=(\bar{R}_{\textnormal{min},i})_{\forall i\in\calN}$, $\bar{\bR}_x=(\bar{R}_{x,i})_{\forall i\in\calN}$, and $\bar{\bR}_y=(\bar{R}_{y,i})_{\forall i\in\calN}$.
Then, Problem~$\Ptwo$ is said to satisfy the time-sharing condition if for any $\bar{\bR}_x$ and $\bar{\bR}_y$, and for any $\theta\in[0,1]$, there always exists a feasible solution $\{\barbp_z, \barbq_z\}$ such that
\begin{equation}\label{eq:time-sharing:condition1}
	\E_{\bh} \left[ R_i(\bp_z^\bh, \bq_z^\bh ; \bh) \right] \ge \theta \bar{R}_{x,i} + (1-\theta) \bar{R}_{y,i}, ~ \forall i\in\calN,
\end{equation}
and
\if\mydocumentclass0
	\begin{align}\label{eq:time-sharing:condition2}
		\E_{\bh} \left[ \sum_{i\in\calN} w_i R_i(\bp_z^{\bh}, \bq_z^{\bh} ; \bh) \right]
		&\ge \theta\E_{\bh} \left[ \sum_{i\in\calN} w_i R_i(\bp_x^{\bh}, \bq_x^{\bh} ; \bh) \right] \nonumber\\
		&\quad+ (1-\theta) \E_{\bh} \left[ \sum_{i\in\calN} w_i R_i(\bp_y^{\bh}, \bq_y^{\bh} ; \bh) \right].
	\end{align}
\else
	\begin{equation}\label{eq:time-sharing:condition2}
		\E_{\bh} \left[ \sum_{i\in\calN} w_i R_i(\bp_z^{\bh}, \bq_z^{\bh} ; \bh) \right] \ge \theta \E_{\bh} \left[ \sum_{i\in\calN} w_i R_i(\bp_x^{\bh}, \bq_x^{\bh} ; \bh) \right] + (1-\theta) \E_{\bh} \left[ \sum_{i\in\calN} w_i R_i(\bp_y^{\bh}, \bq_y^{\bh} ; \bh) \right].
	\end{equation}
\fi
\end{definition}

It has been proven in~\cite{yu2006dual} that if an optimization problem satisfies the time-sharing condition, the strong duality holds regardless of the convexity of the problem.
Hence, we prove Theorem~\ref{thm:zero-duality-gap} by showing that Problem~$\Ptwo$ satisfies the time-sharing condition.
First, for any $\{\barbp_x, \barbq_x\}$ and $\{\barbp_y, \barbq_y\}$, and for any $\theta\in[0,1]$, let us set $\{\bp_z^t, \bq_z^t\}$ to
\begin{equation}\label{eq:appendixA:feasible_solution}
	\{\bp_z^t, \bq_z^t\} = \begin{dcases} \{\bp_x^t, \bq_x^t\} & t \le \lfloor \theta T \rfloor,\\ \{\bp_y^t, \bq_y^t\}, & t \ge \lfloor \theta T+1 \rfloor. \end{dcases}
\end{equation}
Then, the first condition~\eqref{eq:time-sharing:condition1} holds as follows. For all $i\in\calN$,
\if\mydocumentclass0
	\begin{IEEEeqnarray}{rCl}
		\IEEEeqnarraymulticol{3}{l}{%
			\E_{\bh} \left[ R_i(\bp_z^{\bh}, \bq_z^{\bh} ; \bh) \right]
		} \IEEEnonumber*\\
		& = & \lim_{T\to\infty} \frac{1}{T} \sum_{t=1}^{T} R_i(\bp_z^t, \bq_z^t; \bh^t) \\
		& = & \lim_{T\to\infty} \frac{1}{T} \Biggl( \sum_{t=1}^{\lfloor \theta T \rfloor} R_i(\bp_x^t, \bq_x^t; \bh^t) + \sum_{t=\lfloor\theta T+1\rfloor}^{T} R_i(\bp_y^t, \bq_y^t; \bh^t) \Biggr) \\
		& = & \theta \E_{\bh} \left[ R_i(\bp_x^{\bh}, \bq_x^{\bh} ; \bh) \right] + (1-\theta) \E_{\bh} \left[ R_i(\bp_y^{\bh}, \bq_y^{\bh} ; \bh) \right] \\
		& \ge & \theta \bar{R}_{i,x} + (1-\theta) \bar{R}_{i,y}. \IEEEyesnumber
	\end{IEEEeqnarray}
\else
	\begin{IEEEeqnarray}{rCl}
		\E_{\bh} \left[ R_i(\bp_z^{\bh}, \bq_z^{\bh} ; \bh) \right]
		& = & \lim_{T\to\infty} \frac{1}{T} \sum_{t=1}^{T} R_i(\bp_z^t, \bq_z^t; \bh^t) \\
		& = & \lim_{T\to\infty} \frac{1}{T} \Biggl( \sum_{t=1}^{\lfloor \theta T \rfloor} R_i(\bp_x^t, \bq_x^t; \bh^t) + \sum_{t=\lfloor\theta T+1\rfloor}^{T} R_i(\bp_y^t, \bq_y^t; \bh^t) \Biggr) \\
		& = & \theta \E_{\bh} \left[ R_i(\bp_x^{\bh}, \bq_x^{\bh} ; \bh) \right] + (1-\theta) \E_{\bh} \left[ R_i(\bp_y^{\bh}, \bq_y^{\bh} ; \bh) \right] \\
		& \ge & \theta \bar{R}_{i,x} + (1-\theta) \bar{R}_{i,y}. \IEEEyesnumber
	\end{IEEEeqnarray}
\fi
Similarly, the second condition~\eqref{eq:time-sharing:condition2} also holds as follows.
\if\mydocumentclass0
	\begin{IEEEeqnarray}{rCl}
		\IEEEeqnarraymulticol{3}{l}{%
			\E_{\bh} \left[ \sum_{i\in\calN} w_i R_i(\bp_z^{\bh}, \bq_z^{\bh} ; \bh) \right]
		} \IEEEnonumber*\\
		& = & \lim_{T\to\infty} \frac{1}{T} \sum_{t=1}^{T} \sum_{i\in\calN} w_i R_i(\bp_z^t, \bq_z^t; \bh^t) \\
		& = & \lim_{T\to\infty} \frac{1}{T} \left( \sum_{t=1}^{\lfloor \theta T \rfloor} \sum_{i\in\calN} w_i R_i(\bp_x^t. \bq_x^t; \bh^t) \vphantom{+ \sum_{t=\lfloor\theta T+1\rfloor}^{T} \sum_{i\in\calN} w_i R_i(\bp_y^t, \bq_y^t; \bh^t)} \right. \\
		&   & \qquad\qquad\qquad \vphantom{\sum_{t=1}^{\lfloor \theta T \rfloor} \sum_{i\in\calN} w_i R_i(\bp_x^t. \bq_x^t; \bh^t)} \left. +\> \sum_{t=\lfloor\theta T+1\rfloor}^{T} \sum_{i\in\calN} w_i R_i(\bp_y^t, \bq_y^t; \bh^t) \right) \\
		& = & \theta \E_{\bh} \left[ \sum_{i\in\calN} w_i R_i(\bp_x^{\bh}, \bq_x^{\bh}; \bh) \right] \\
		&   & \qquad\qquad\qquad +\> (1-\theta) \E_{\bh} \left[ \sum_{i\in\calN} w_i R_i(\bp_y^{\bh}, \bq_y^{\bh}; \bh) \right]. \IEEEyesnumber\IEEEeqnarraynumspace
	\end{IEEEeqnarray}
\else
	\begin{IEEEeqnarray}{rCl}
		\E_{\bh} \left[ \sum_{i\in\calN} w_i R_i(\bp_z^{\bh}, \bq_z^{\bh} ; \bh) \right]
		& = & \lim_{T\to\infty} \frac{1}{T} \sum_{t=1}^{T} \sum_{i\in\calN} w_i R_i(\bp_z^t, \bq_z^t; \bh^t) \\
		& = & \lim_{T\to\infty} \frac{1}{T} \left( \sum_{t=1}^{\lfloor \theta T \rfloor} \sum_{i\in\calN} w_i R_i(\bp_x^t. \bq_x^t; \bh^t) + \sum_{t=\lfloor\theta T+1\rfloor}^{T} \sum_{i\in\calN} w_i R_i(\bp_y^t, \bq_y^t; \bh^t) \right) \IEEEeqnarraynumspace\\
		& = & \theta \E_{\bh} \left[ \sum_{i\in\calN} w_i R_i(\bp_x^{\bh}, \bq_x^{\bh}; \bh) \right] + (1-\theta) \E_{\bh} \left[ \sum_{i\in\calN} w_i R_i(\bp_y^{\bh}, \bq_y^{\bh}; \bh) \right]. \IEEEyesnumber
	\end{IEEEeqnarray}
\fi
Consequently, we can conclude that the time-sharing condition holds in Problem~$\Ptwo$, resulting in the strong duality. \qed

\fi

\bibliographystyle{IEEEtran}
\bibliography{IEEEabrv,Scheduling_NOMA}

\end{document}